\documentclass[11pt]{article}
\usepackage[utf8]{inputenc}

\usepackage{lipsum}
\usepackage{verbatim}
\usepackage{caption}
\usepackage{amssymb}
\usepackage{amsmath}
\usepackage{tabularx}
\usepackage{color}
\usepackage{stmaryrd}
\usepackage{multirow}
\usepackage{mathtools}
\usepackage{hyperref}
\usepackage{epsfig,latexsym}
\usepackage[resetlabels]{multibib}

\newcites{supp}{Supplementary references}

\let\oldequation\equation
\let\oldendequation\endequation

\usepackage{amsmath,amssymb,amsthm,graphicx,verbatim,bm,fancyhdr,bbm,enumerate,color, multirow, tabularx,mathtools}
\usepackage{hyperref}

\DeclareMathAlphabet{\mathpzc}{OT1}{pzc}{m}{it}  

\theoremstyle{definition}
\theoremstyle{plain}
\newtheorem{theorem}{Theorem}

\newtheorem*{proposition*}{Proposition}
\newtheorem{lemma}{Lemma}

\theoremstyle{remark}

\newtheorem*{remark*}{Remark}

\newtheorem*{terminology*}{Terminology}
\newtheorem*{notation*}{Notation}

\newcommand{\E}{\mathbb{E}}

\newcommand{\inquotes}[1]{``#1''}

\usepackage{lineno}

\let\oldalign\align
\let\oldendalign\endalign

\renewenvironment{align}
  {\linenomathNonumbers\oldalign}
  {\oldendalign\endlinenomath}

\let\oldequation\equation
\let\oldendequation\endequation

\renewenvironment{equation}
  {\linenomathNonumbers\oldequation}
  {\oldendequation\endlinenomath}

\title{
Optimal dosing of anti-cancer treatment \\ 
under 
drug-induced plasticity \\[12pt]
}

\author{Einar Bjarki Gunnarsson$^{1,2,*}$ \and\hspace*{-6pt} Benedikt Vilji Magnússon$^{1}$ \and\hspace*{-6pt} Jasmine Foo$^{2}$}
\date{%
    \footnotesize $^1$Division of Applied Mathematics, Science Institute, University of Iceland \\[0pt]
    $^2$School of Mathematics, University of Minnesota \\[-3pt]
    * corresponding author (ebg@hi.is)
    } 

\renewenvironment{abstract}
 {\small
  \begin{center}
  \bfseries \abstractname\vspace{0pt}\vspace{0pt}
  \end{center}
  \list{}{%
    \setlength{\leftmargin}{21mm}
    \setlength{\rightmargin}{\leftmargin}%
  }%
  \item\relax}
 {\endlist}

\begin{document}

\maketitle

\begin{abstract}
While cancer has traditionally been considered a genetic disease, mounting evidence indicates an important role for non-genetic (epigenetic) mechanisms. Common anti-cancer drugs have recently been observed to induce the adoption of non-genetic drug-tolerant cell states, thereby accelerating the evolution of drug resistance. This confounds conventional high-dose treatment strategies aimed at maximal tumor reduction, since high doses can simultaneously promote non-genetic resistance. In this work, we study optimal dosing of anti-cancer treatment under drug-induced cell plasticity. We show that the optimal dosing strategy steers the tumor to a fixed equilibrium composition between sensitive and tolerant cells, while precisely balancing the trade-off between cell kill and tolerance induction. The optimal equilibrium strategy ranges from applying a low dose continuously to applying the maximum dose intermittently, depending on the dynamics of tolerance induction. We finally discuss how our approach can be integrated with {\em in vitro} data to derive patient-specific treatment insights.
\\
\end{abstract}


\section{Introduction}

Cancer is the result of a complex evolutionary process which 
results in cells gaining the ability to divide uncontrollably \cite{merlo2006cancer,greaves2015evolutionary}.
While genetic mutations have long been understood to be important drivers of cancer evolution, it has more recently become clear that epigenetic mechanisms such as methylation and acetylation are sufficient to fulfill the traditional hallmarks of cancer \cite{flavahan2017epigenetic,jones2007epigenomics,Brock2009,brown2002epigenomics}.
Epigenetic mechanisms are both heritable and reversible, and they can enable cells to switch 
dynamically
between two or more phenotypic states, which commonly show differential responses to drug treatment.
In fact, tumor cells in many cancer types have been observed to 
adopt reversible drug-tolerant or stem-like states {\em in vitro}, 
which enables the tumor to temporarily evade drug treatment and
sets the stage for the evolution of more permanent resistance mechanisms \cite{sharma2010chromatin,goldman2015temporally,ramirez2016diverse,Engelman2016,shaffer2017rare,su2017single,neftel2019integrative,gunnarsson2020understanding}.
While these drug-tolerant states can exist independently of treatment \cite{sharma2010chromatin,pisco2013non,shaffer2017rare}, recent evidence indicates that common anti-cancer drugs can also induce or accelerate their adoption,
a phenomenon we refer to as {\em drug-induced plasticity}
\cite{pisco2013non,goldman2015temporally,shaffer2017rare,su2017single,vipparthi2022emergence,russo2022modified}.
Drug induction
is already well known in the context of genetic mutations,
where 
both anti-cancer and 
anti-bacterial 
drugs have been observed to 
cause 
genomic instability or 
{\em mutagenesis},
referring to 
an elevated mutation rate 
occurring as a result of 
the stress imposed by the treatment
\cite{szikriszt2016comprehensive,venkatesan2017treatment,russo2019adaptive,cipponi2020mtor}.

Anti-cancer drug treatment has traditionally been administered under the \inquotes{maximum tolerated dose} (MTD) paradigm, where the goal is to eradicate the tumor as quickly as possible and minimize the probability of spontaneous resistance-conferring mutations
\cite{kerbel2004anti,norton1986norton}.
Non-genetic cell plasticity 
fundamentally
confounds the reasoning behind this strategy,
especially if it is induced by the drug.
First, if drug-tolerant states are adopted independently of the drug, it is likely that drug-tolerant cells will preexist treatment in large numbers, even if they form only a small fraction of the tumor bulk \cite{Brock2009}. 
In addition, since non-genetic mechanisms usually operate much faster than genetic mutations \cite{Brock2009,jones2007epigenomics,brown2014poised}, a non-negligible fraction of drug-sensitive cells will adopt drug tolerance under treatment.
This means that
it can be impossible to kill the tumor,
no matter how large a dose is applied \cite{gunnarsson2020understanding}.
Second, if the drug induces 
tolerance adoption
in a dose-dependent manner,
applying larger doses with the aim of maximizing cell kill
has the downside of promoting resistance evolution.
This trade-off has been explored in several recent mathematical modeling works, which generally find that continuous low-dose or intermittent strategies outperform MTD treatment \cite{greene2019mathematical,akhmetzhanov2019modelling,greene2020mathematical,kuosmanen2021drug,angelini2022model,corigliano2025optimal}.
However, these works usually focus on one potential model of drug-induced tolerance and some only consider the case of irreversible genetic resistance.
In addition, many of these works only consider constant or intermittent strategies.
Overall, there is a lack of generalized insights into
how different potential forms of drug-induced tolerance influence 
optimal treatment decisions.

\begin{figure}
    \centering
    \includegraphics[scale=0.9]{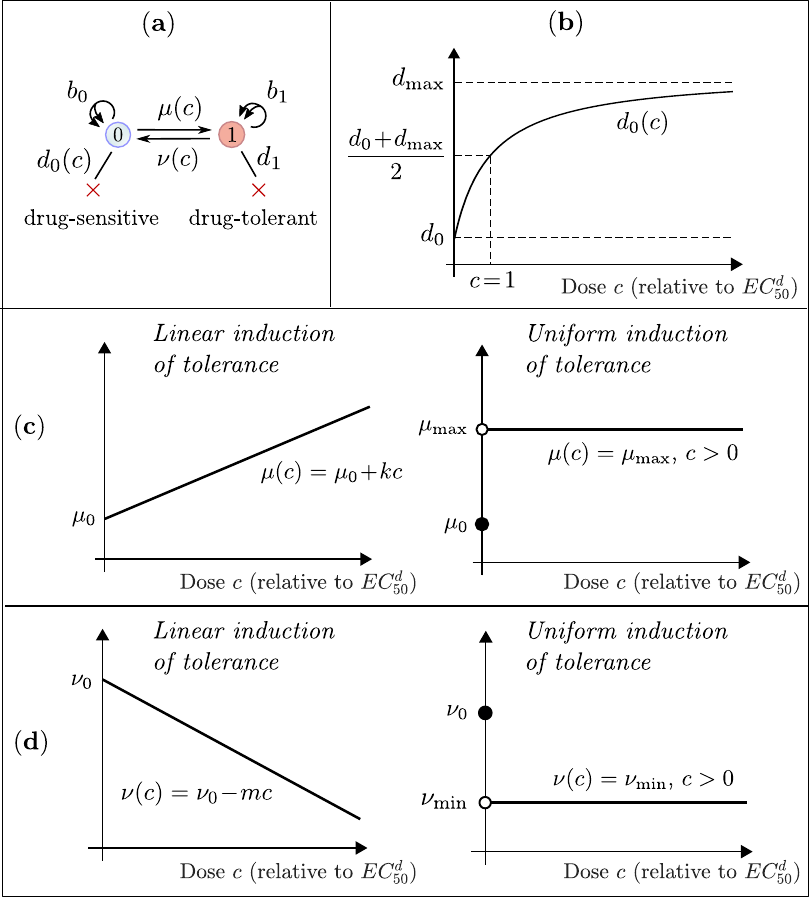}
    \caption{{\bf Mathematical model of drug-induced tolerance.}
    {\bf (a)} In the mathematical model, cells transition between two states, a drug-sensitive (type-0) state and a drug-tolerant (type-1) state.
    Sensitive cells divide at rate $b_0$, die at rate $d_0$ and transition to tolerance at rate $\mu$.
    Tolerant cells divide at rate $b_1$, die at rate $d_1$ and transition to sensitivity at rate $\nu$.
    The sensitive cell death rate $d_0$ and the transition rates $\mu$ and $\nu$ 
    depend on the drug dose $c$ (Section \ref{sec:model}).
    {\bf (b)} The sensitive cell death rate follows a Michaelis-Menten equation of the form $d_0(c) = d_0 + (d_{\rm max}-d_0) c/(c+1)$, where $d_0$ is the death rate in the absence of drug and 
    $d_{\rm max}$ is the saturation death rate under an arbitrarily large drug dose.
    The dose $c$ is normalized to the ${\rm EC}_{50}$ dose for $d_0$, meaning 
    that the drug has half the maximal effect at dose $c=1$ (Section \ref{sec:drugeffectprolif}).
    {\bf (c)} The transition rate $\mu(c)$ from drug-sensitivity to drug-tolerance is assumed either a linearly increasing function of  $c$ or to be uniformly elevated in the presence of drug.
    These two forms of drug-induced tolerance were observed in a  recent study by Russo et al.~\cite{russo2022modified} and they also emerge as limiting cases of more general Michaelis-Menten dynamics (Section \ref{sec:drugeffectswitching}).
    {\bf (d)} The transition rate $\nu(c)$ from drug-tolerance to drug-sensitivity is similarly assumed either a linearly decreasing function of $c$ or uniformly inhibited in the presence of drug (Section \ref{sec:drugeffectswitching}).
    }
        \label{fig:model_overviews}
\end{figure}

In this work, 
we seek 
a more systematic understanding of 
optimal dosing strategies
under 
drug-induced cell plasticity.
We seek both
a qualitative understanding of the
characteristics of optimal strategies 
and 
simple quantitative methods for computing them.
We consider many different possible forms of drug-induced tolerance and explore how and why different forms lead to different optimal strategies.
Our work differs from previous work 
in several important ways.
First, whereas prior work generally assumes that tolerance
is induced through elevated transitions from drug-sensitivity to drug-tolerance, it is also biologically plausible that the drug inhibits 
transitions from tolerance back to sensitivity,
or that it affects both transitions simultaneously.
For example, hypermethylation in the promoter regions of tumor suppressor genes has been associated with gene silencing and drug resistance in many cancers \cite{merlo2006cancer,jones2007epigenomics,esteller2002cpg,romero2020role}.
When a cell divides, existing methylation patterns are preserved by the maintence methyltransferase DNMT1, while {\em de novo} methylation is carried out by the methyltransferases DNMT3A/3B \cite{okano1999dna,esteller2008epigenetics}.
If a drug-tolerant state is conferred via methylation, 
drug-induced upregulation of DNMT3A/3B will 
induce the adoption of the tolerant state, while drug-induced upregulation of DNMT1 will 
inhibit transitions out of it.
The notion that the drug may not only elevate transitions 
from sensitivity to tolerance is further supported by 
experimental evidence indicating that the drug can
influence all possible phenotypic transitions simultaneously
\cite{goldman2015temporally,su2017single,vipparthi2022emergence}.
Second, whereas many prominent prior works assume that the level of tolerance induction is linear as a function of the drug dose, we also consider the case where tolerance is induced uniformly 
in the presence of drug.
These two forms of drug-induced tolerance were observed in a recent experimental work by Russo et al.~\cite{russo2022modified}, and they emerge as limiting cases of a more general Michaelis-Menten form of tolerance induction as is discussed in 
Section \ref{sec:drugeffectswitching}.
Third, while most previous work is computational in nature,
we establish mathematical results which provide generalized insights 
that apply across all biologically relevant parameter regimes. 
Finally, whereas some previous works focus on 
irreversible transitions from drug-sensitivity to drug-resistance, 
we study the case of reversible phenotypic switching between sensitivity and tolerance.
We compare and contrast our results with previous work in more detail in Section \ref{sec:previouswork}.

\section{Results}

\begin{figure*}
    \centering
    \includegraphics[scale=1]{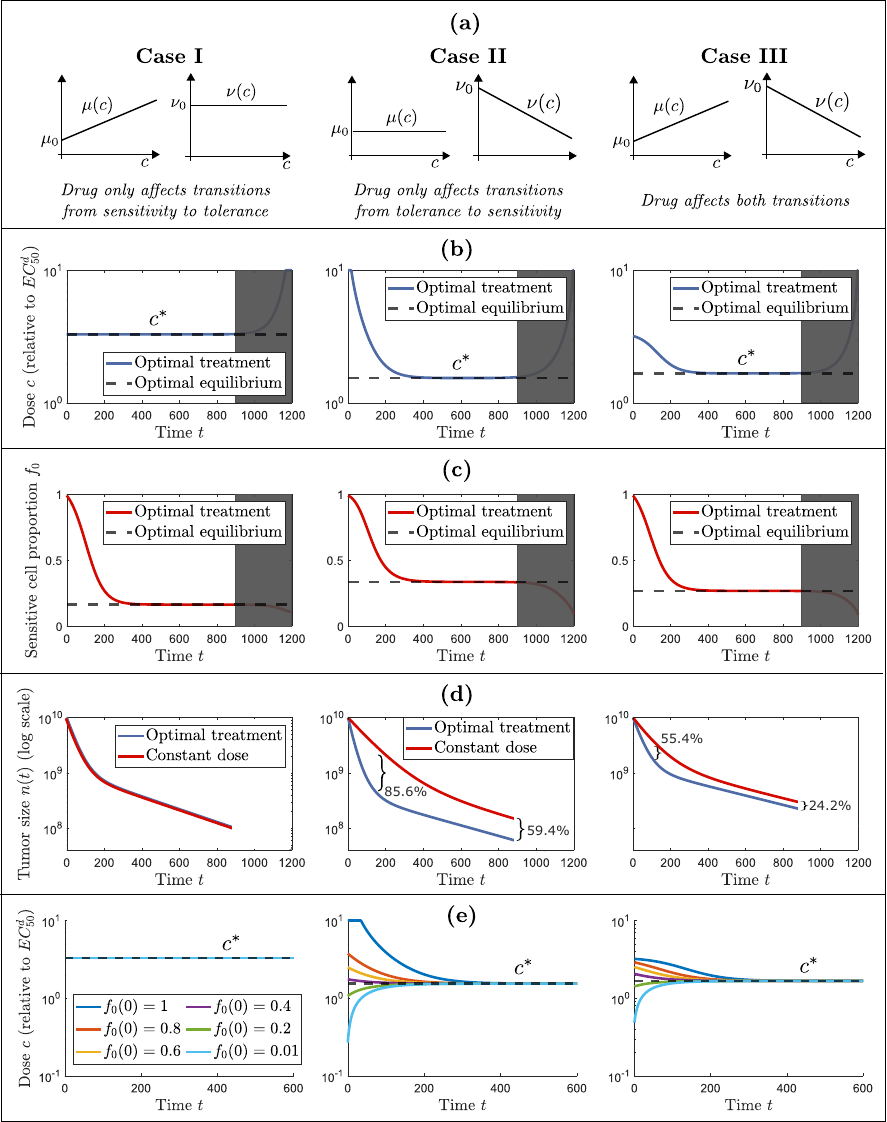}
\end{figure*}

\begin{figure*}
        \caption{
    {\bf Optimal dosing under linear induction of tolerance.}
    {\bf (a)}
In our investigation of linear induction of tolerance, we consider three cases.
In Case I, the drug only increases the transition rate from sensitivity to tolerance ($\mu(c) = \mu_0 + kc$ and $\nu(c) = \nu_0$), 
in Case II, the drug only decreases the transition rate from tolerance to sensitivity 
($\mu(c) = \mu_0$ and $\nu(c) = \nu_0 - m c$), 
and in Case III, the drug does both 
($\mu(c) = \mu_0 + kc$ and $\nu(c) = \nu_0 - m c$).
    {\bf (b)} Optimal dosing strategies under linear induction of tolerance
    computed using the forward-backward sweep method.
    In all cases, the optimal dosing strategy involves both a transient phase and an equilibrium phase.
    During the equilibrium phase, a constant dose $c^\ast$ is applied,
    which minimizes the long-run tumor growth rate $\sigma(c)$ under a constant dose $c$, as indicated by the dotted lines.
    The greyed-out part of each dosing strategy is a boundary effect which arises due to 
    the optimization goal of minimizing the tumor size at a specific time $T$,
    with no regard for what is best in the long run.
    If we were to extend the treatment horizon beyond time $T$, the 
    period applying the constant dose $c^\ast$
    would be extended (Figure \ref{fig:S1}).
    We therefore ignore this boundary effect
    and focus on what is optimal over a longer time horizon.
    {\bf (c)} Proportion of sensitive cells $f_0(t)$ over time under the optimal dosing strategy. The constant equilibrium dose $c^\ast$ maintains the tumor at a fixed proportion between sensitive and tolerant cells.
    {\bf (d)} Tumor size evolution under the optimal strategy compared with the tumor size evolution under the constant-dose strategy applying $c^\ast$ throughout.
    For Case I, 
    the two strategies are the same.
    For Cases II/III, the optimal strategy results in 85.6\%/55.4\% greater tumor reduction during the initial stages of treatment, and 59.4\%/24.2\% greater reduction long-term.
        {\bf (e)} If the proportion of sensitive cells $f_0(0)$ at the start of treatment is changed, 
        only the transient phase of the optimal treatment is affected, whereas the optimal equilibrium dose $c^\ast$ remains the same.
    For Case I, the transient phase is independent of the initial condition.
    }
            \label{fig:optimallinear}
\end{figure*}

\subsection{Mathematical model of drug-induced tolerance enables derivation of optimal dosing strategies}

We propose a mathematical model of a tumor undergoing anti-cancer treatment,
in which cells are able to transition between two cell states, a drug-sensitive (type-0) and a drug-tolerant (type-1) state.
Cells of type-$i$ divide at rate $b_i$ and die at rate $d_i$, with net division rate $\lambda_i = b_i-d_i$.
Each drug-sensitive cell is assumed to transition to the drug-tolerant state at rate $\mu$, and each tolerant cell transitions back to the sensitive state at rate $\nu$ (Figure \ref{fig:model_overviews}(a)).
To model the effect of the drug on cell proliferation, we assume that the sensitive cell death rate $d_0(c)$ increases with the drug dose $c$ (Figure \ref{fig:model_overviews}(b)).
Equivalently, it can be assumed that the drug decreases the division rate of sensitive cells (Section \ref{sec:drugeffectprolif}).
To model drug-induced tolerance, the transition rates $\mu(c)$ and $\nu(c)$ between sensitivity and tolerance are assumed to depend on the drug dose $c$, either 
in a linear or uniform fashion
(Figure \ref{fig:model_overviews}(c)-(d)). 
A detailed description of the mathematical model is provided in Sections \ref{sec:model}--\ref{sec:drugeffectswitching}.

The time evolution of the number of drug-sensitive (type-0) and drug-tolerant (type-1) cells 
can be described by the differential equations
\begin{align*}
        & \frac{dn_0}{dt} = (\lambda_0(c)-\mu(c))n_0+\nu(c)n_1, \\
    & \frac{dn_1}{dt} = (\lambda_1-\nu(c))n_1+\mu(c)n_0,
\end{align*}
where the drug dose is a function of time, $c(t)$ (Section \ref{sec:systemeqns}).
The proportion of drug-sensitive cells over time, $f_0(t) = n_0(t)/(n_0(t)+n_1(t))$, then follows the differential equation
\[
\frac{df_0}{dt} = \big(\lambda_1-\lambda_0(c)\big)f_0^2 - \big(\lambda_1-\lambda_0(c)+\mu(c)+\nu(c)\big)f_0+\nu(c).
\]
If a constant dose is applied,
$c(t)=c$ for all $t \geq 0$, 
the tumor eventually grows exponentially at a fixed rate $\sigma(c)$ given by
\begin{align} \label{eq:sigmaexplicit0}
     \sigma(c) = \frac12 \Big(&\lambda_0(c)-\mu(c) + \lambda_1-\nu(c) + \sqrt{(\lambda_0(c)-\mu(c)-\lambda_1+\nu(c))^2+4\mu(c)\nu(c)}\Big).
\end{align}
Moreover, 
the intratumor composition eventually reaches an equilibrium,
where the proportion of sensitive cells 
in the population 
becomes fixed at
\begin{align} \label{f0barcexplicit}
    \bar{f}_0(c) = \frac{\nu(c)}{\sigma(c)-\lambda_0(c)+\mu(c)+\nu(c)}.
\end{align}
Thus, under a constant dose $c$, it is possible  to maintain the tumor at a stable intratumor composition, which leads to a fixed exponential growth rate $\sigma(c)$ (Section \ref{sec:constantdoseequilibrium}).

We consider a dosing strategy $(c(t))_{t \in [0,T]}$ optimal if it
minimizes the
total number of cells, $n_0(T)+n_1(T)$, at the end of a finite time horizon $[0,T]$.
When analyzing 
which dosing strategy is optimal 
in the long run,
we also consider 
an infinite-horizon problem where the goal is to minimize the long-run average tumor growth rate (Section \ref{sec:optcontrolproblem}).
We can show that to determine the optimal strategy,
it is sufficient to consider its effect on the 
proportion of sensitive cells $f_0(t)$ over time, 
as opposed to its effect on the individual 
subpopulation counts $n_0(t)$ and $n_1(t)$ (Sections \ref{sec:systemeqns} and \ref{sec:optcontrolproblem}).
This enables us to reduce an optimal control problem involving two state variables $n_0(t)$ and $n_1(t)$ to a single-state-variable problem,
which simplifies our computational algorithms
and especially our mathematical analysis.
In line with this insight, we use the proportion of sensitive cells over time to demonstrate the effects of optimal dosing strategies on tumor evolution in the following sections.

\subsection{Constant low dosing is optimal in the long run under linear induction of tolerance} \label{sec:optlinearinduction}

We now investigate optimal dosing strategies under linear induction of tolerance.
We consider three cases, depicted in Figure \ref{fig:optimallinear}(a), depending on whether the drug only affects transitions from sensitivity to tolerance (Case I), it only affects transitions from tolerance to sensitivity (Case II), or it affects both transitions simultaneously (Case III).
In Figure \ref{fig:optimallinear}(b), we show the optimal strategy for each case, computed using our own implementation of the forward-backward sweep method in MATLAB (Section \ref{app:fbsm}).
For Case I, the optimal strategy applies a constant low dose $c^\ast$ from the beginning.
For Case II, the optimal strategy starts with the maximum dose $c_{\rm max}$, and then gradually decreases it to a constant low dose $c^\ast$.
Under Case II, higher doses can be used at the onset of treatment to reduce the predominantly sensitive tumor more quickly, without the downside of promoting transitions from sensitivity to tolerance.
For Case III, the optimal strategy has characteristics of both previous cases, as it does apply larger doses in the beginning, but it does not start with the maximum dose.

By combining these results with mathematical analysis, 
we can derive several general treatment insights.
First, the optimal dosing strategy always
involves both a transient phase and an equilibrium phase.
During the equilibrium phase, a constant dose $c^\ast$ is applied, which maintains the tumor at a fixed proportion between sensitive and tolerant cells (Figure \ref{fig:optimallinear}(c)).
This is true irrespective of the model parameter values (Section \ref{app:proofs}).
To demonstrate the importance of distinguishing these two treatment phases,
we compare in Figure \ref{fig:optimallinear}(d)
the optimal strategy with the constant-dose strategy 
applying the equilibrium dose $c^\ast$ throughout.
Eventually, the tumor reduces at the same rate 
under both strategies.
However, during the transient phase, the optimal strategy reduces the tumor much faster
for Cases II and III, and this difference in tumor size is 
maintained throughout the course of treatment.

Second, we can prove mathematically that the optimal equilibrium dose $c^\ast$ minimizes the long-run tumor growth rate $\sigma(c)$ under a constant dose $c$ (Section \ref{app:proofs}).
This means that the explicit expression \eqref{eq:sigmaexplicit0} for $\sigma(c)$ has clinical value on its own, as it can be used to address questions concerning ultimate treatment success.
For example, if for a particular cancer type and a particular 
drug, we are interested in knowing whether there exists {\em any} dosing schedule capable of driving long-run tumor reduction (however complicated), we simply need to check whether there exists a constant dose $c$ such that $\sigma(c)<0$.
We note that since the equilibrium dose $c^\ast$ is determined by long-run tumor growth dynamics, it is not affected by changes in the initial proportion between sensitive and tolerant cells
(Figure \ref{fig:optimallinear}(e)).

Third, we can say that the goal of the transient phase is to 
steer the tumor to a desired composition,
and the goal of the equilibrium phase is to maintain it at that composition.
We can show that 
during the transient phase,
once the proportion of sensitive cells reaches some value $f_0$,
the optimal dose $c$ to give at that moment must maximize
\begin{align} \label{eq:transientminimize}
    \frac{\sigma(c^\ast)-u(c,f_0)}{-f_0'},
\end{align}
where
$u(c,f_0) = (\lambda_0(c)-\lambda_1)f_0+\lambda_1$ is the instantaneous tumor growth rate (Section \ref{app:transient}).
Applying the maximum dose $c_{\rm max}$ would kill the sensitive cells as quickly as possible,
which would lead to the lowest possible growth rate $u(c,f_0)$ 
and thus the largest possible 
numerator
in \eqref{eq:transientminimize}.
However, the maximum dose would simultaneously induce a large number of sensitive cells to adopt tolerance, which would lead to a large reduction $-f_0'$ in the proportion of sensitive cells in the denominator of \eqref{eq:transientminimize}.
Overall, expression \eqref{eq:transientminimize} describes precisely how the optimal strategy must balance the trade-off between sensitive cell kill and tolerance induction as it steers the 
tumor
to the optimal equilibrium composition.

Finally, our mathematical analysis suggests a simple and efficient method to compute the optimal dosing strategy, which simultaneously provides an insight into how the doses are selected.
First, the optimal equilibrium dose $c^\ast$  can be derived by minimizing expression \eqref{eq:sigmaexplicit0}, and the associated equilibrium proportion of sensitive cells $\bar{f}_0(c^\ast)$ can be computed using expression \eqref{f0barcexplicit}.
Then, the transient phase can be computed by iteratively determining the maximizing dose $c$ in \eqref{eq:transientminimize} and updating the value of $f_0$, as visualized in Figure \ref{fig:transient_steering}(a).
In Figure \ref{fig:transient_steering}(b), we confirm that the transient phase computed using this approach agrees with the transient phase shown in Figure \ref{fig:optimallinear} (obtained using the forward-backward sweep method).

\begin{figure}
    \centering
    \includegraphics[scale=1]{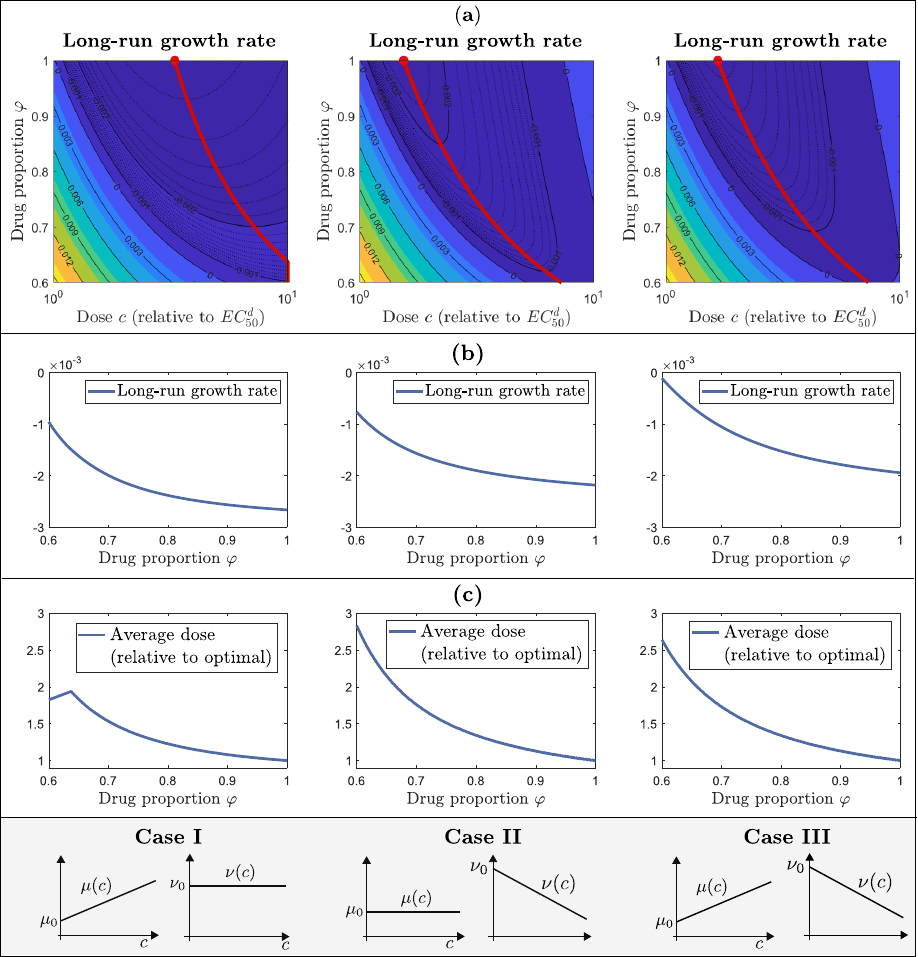}
    \caption{
    {\bf Continuous and pulsed schedules under linear induction of tolerance.}
    {\bf (a)} Long-run average growth rate of the tumor
    under a range of possible pulsed schedules of the form $(c,\varphi)$, 
    where $c$ is the drug dose applied and 
    $\varphi$ is the proportion of drug exposure.
    The red curves show for each exposure proportion $\varphi$
    the optimal dose $c^\ast(\varphi)$ to give.
    The red dots show the optimal pulsed schedules, which are continuous for Cases I, II and III, consistent with our results in 
    Figure \ref{fig:optimallinear}.
    {\bf (b)}
    Long-run average growth rate
    under the optimal dose $c^\ast(\varphi)$, shown as a function of the proportion $\varphi$ of drug exposure.
    {\bf (c)} Average dose applied during each treatment cycle under the optimal dose $c^\ast(\varphi)$, shown as a function of the proportion $\varphi$ of drug exposure, normalized to the average dose under the optimal schedule.
    }
    \label{fig:pulsedlinear}
\end{figure}

\subsection{Long treatment breaks significantly deterioriate treatment efficacy}
\label{sec:pulsedlinear}

Under linear induction of tolerance, it is optimal in the long run to expose the tumor continuously to the drug at a majority drug-tolerant composition.
We are next interested in investigating how the optimal treatment compares with pulsed treatments,
where the tumor is taken periodically off drug to allow tolerant cells to revert back to sensitivity.
This comparison is made using the long-run average growth rate of the tumor 
(Section \ref{sec:optcontrolproblem}).

We consider a pulsed schedule as a triple $(c,t_{\rm cycle},\varphi)$ encoding the drug dose $c$, the duration of each treatment cycle $t_{\rm cycle}$, and the proportion of time $\varphi$ the drug is applied during each cycle.
A constant schedule is considered a pulsed schedule with $\varphi=1$.
It is relatively simple to show that in our model, for any pulsed schedule $(c,t_{\rm cycle},\varphi_1)$, there exists a schedule $(c,t_{\rm cycle}/2,\varphi_2)$ with half the cycle length which yields the same or lower 
long-run growth rate (Section \ref{app:pulsed}).
Therefore, when investigating pulsed schedules, it is sufficient to consider idealized schedules with an arbitrarily short cycle length.
However, it is also simple to show that once the cycle length has become sufficiently small, the long-run average growth rate becomes effectively invariant to the cycle length (Section \ref{app:pulsed}).
Thus, even if the best pulsed schedule involves switching arbitrarily fast between drug and no drug, a comparable 
outcome can be obtained using a much longer cycle time.
This is an important insight, as it means that the model-recommended pulsed schedule is more clinically feasible.

For a given pulsed schedule $(c,\varphi)$ with an arbitrarily short cycle length, the long-run average growth rate of the tumor can be computed using the relatively simple formula \eqref{eq:sigmapulsed}.
In Figure \ref{fig:pulsedlinear}(a), we show the long-run growth rate for a wide range of possible pulsed schedules $(c,\varphi)$ under Cases I, II and III.
The red lines show for each proportion $\varphi$ of drug exposure the optimal dose $c^\ast(\varphi)$ to administer, and the red dots indicate the optimal schedules.
As expected, a constant low dose is optimal in the long run.
Figure \ref{fig:pulsedlinear}(b)--(c) shows that 
for all three cases, the optimal schedule is quite sensitive to the duration of drug exposure.
Instituting treatment breaks of 20--30\% or longer results in a significantly larger tumor growth rate and a larger average dose being required to achieve that growth rate.
From a clinical perspective, understanding to what extent 
the optimal strategy can be perturbed to insert treatment breaks
can be useful for 
reducing potential side effects of continuous drug exposure.
Plots like those in Figure \ref{fig:pulsedlinear} could help a clinician understand the implications of instituting treatment breaks of a certain length, and to determine which dose is best for that length.

\subsection{Pulsed maximum dosing is optimal in the long run under uniform induction of tolerance}
\label{sec:uniforminduction}

We now consider the case where drug tolerance is induced in a uniform fashion, 
meaning that the drug effect is constant whenever the drug is present
(Section \ref{sec:drugeffectswitching}).
In this case, the forward-backward sweep method does not converge to a unique solution, due to the discontinuity of the transition rate functions $\mu(c)$ and $\nu(c)$.
We therefore rely on mathematical analysis.

First, 
we can show that in the long run, it is optimal to apply a pulsed schedule alternating arbitrarily quickly between the maximum dose $c_{\rm max}$ and no dose (Section \ref{app:proofs}).
To verify this result computationally, we compare in Figure \ref{fig:optimalheaviside}(b) the long-run average tumor growth rate over a wide range of possible pulsed schedules $(c,\varphi)$.
For each given proportion of drug exposure $\varphi$, it is optimal to apply the maximum dose $c_{\rm max}$ whenever the drug is present, as indicated by the red lines.
For Case I, it is optimal to apply the maximum dose continuously, and for Cases II and III, it is optimal to alternate between the maximum dose and no dose, as indicated by the red dots.
Thus, whereas low continuous dosing is optimal in the long run for linear induction of tolerance, for uniform induction, it is optimal to apply the maximum dose either continuously or intermittently.
This result is intuitive since under uniform induction, once any drug dose is applied, there is no downside in applying a larger dose.

\begin{figure*}
    \includegraphics[scale=1]{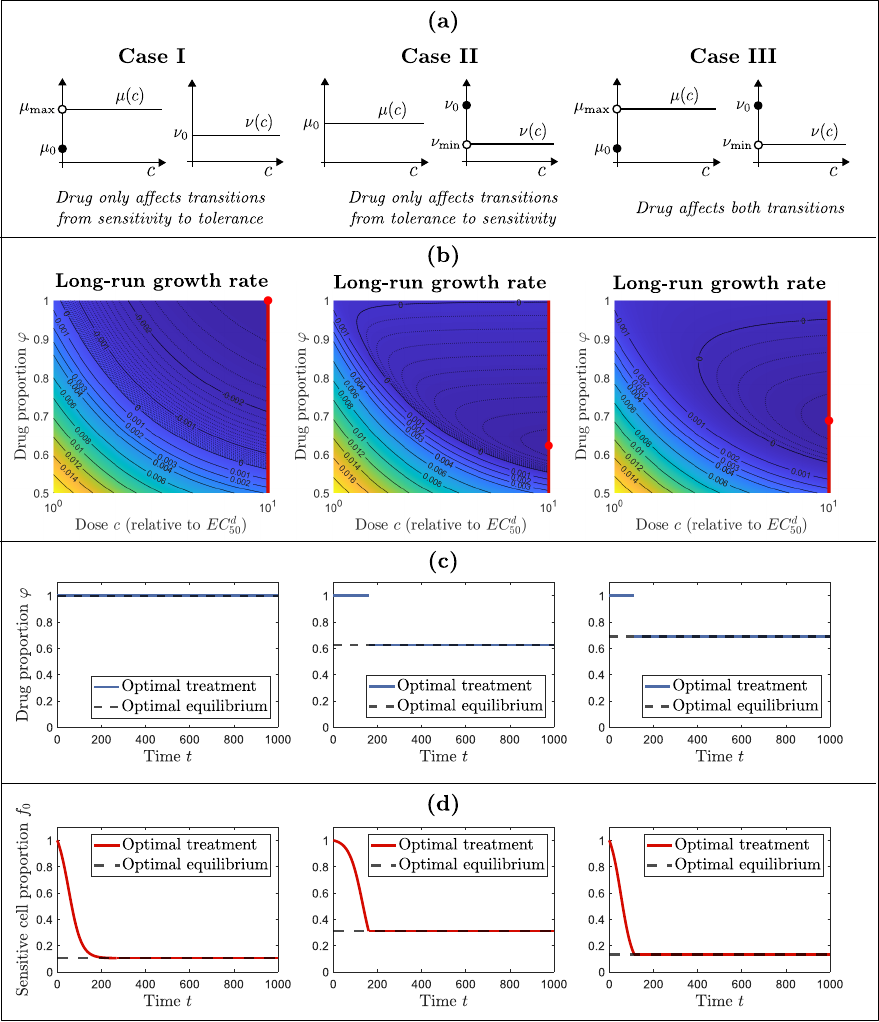}
\end{figure*}

\begin{figure}
    \centering
    \caption{{\bf Optimal dosing under uniform induction of tolerance.}
        {\bf (a)}
In our investigation of uniform induction of tolerance, we consider three cases.
In Case I, the drug only increases the transition rate from sensitivity to tolerance ($\mu(0)=\mu_0$ and $\mu(c) = \mu_{\rm max}$ for $c>0$, $\nu(c) = \nu_0$ for $c \geq 0$), 
in Case II, the drug only decreases the transition rate from tolerance to sensitivity 
($\mu(c) = \mu_0$ for $c \geq 0$, $\nu(0)=\nu_0$ and $\nu(c) = \nu_{\rm min}$ for $c>0$), 
and in Case III, the drug does both
($\mu(0)=\mu_0$ and $\mu(c) = \mu_{\rm max}$ for $c>0$, $\nu(0)=\nu_0$ and $\nu(c) = \nu_{\rm min}$ for $c>0$).
     {\bf (b)} 
    The heatmaps show the long-run average tumor growth rate 
    for a range of possible pulsed schedules of the form $(c,\varphi)$
    for Cases I, II and III.
    The red curves show each for exposure proportion $\varphi$
    the optimal dose $c^\ast(\varphi)$ to give,
    and the red dots show the optimal pulsed schedules.
    For each possible exposure proportion $\varphi$, it is
    optimal to apply the maximum dose $c_{\rm max}$.
    For Case I, where the drug only affects transitions from sensitivity to tolerance, it is optimal to apply the maximum dose continuously.
    For Cases II and III, 
    it is optimal to alternate between the maximum dose and no dose.
    {\bf (c)} Under uniform induction of tolerance, 
    we can prove that during the transient phase, 
    it is optimal to apply the maximum dose throughout,
    and during the equilibrium phase,
    it is optimal to apply a pulsed schedule $(c_{\rm max},\varphi)$ alternating between the maximum allowed dose $c_{\rm max}$ and no dose (possibly with $\varphi =1$).
    Examples of optimal strategies are shown for Cases I, II and III.
    For Case I, where the drug only increases transitions from sensitivity to tolerance, it is optimal to apply the maximum dose continuously throughout the entire treatment period.
    {\bf (d)} Proportion of sensitive cells $f_0(t)$ over time under the optimal dosing strategy. Same as for linear induction of tolerance, the optimal dosing strategy maintains the tumor at a fixed proportion between sensitive and tolerant cells in the long run.
    }
    \label{fig:optimalheaviside}
\end{figure}

Second, by analyzing the transient phase using the same approach as for linear induction of tolerance, 
we can show that the optimal transient strategy is to apply the maximum dose $c_{\rm max}$ continuously under Cases I, II and III (Section \ref{app:transient}).
This enables us to derive
optimal dosing strategies for all cases under uniform induction of tolerance (Figure \ref{fig:optimalheaviside}(c)).
For Case I, it is optimal to apply the maximum dose $c_{\rm max}$ continuously from the beginning.
This implies that drug-induced tolerance does not necessarily call for 
low or intermittent dosing.
For Cases II and III, it is optimal to start with the maximum dose and then alternate between the maximum dose and no dose.
For all cases, the optimal strategy eventually maintains the tumor at a fixed composition between sensitive and tolerant cells (Figure \ref{fig:optimalheaviside}(d)).
Thus, same as for the case of linear induction, the optimal strategy involves a transient and an equilibrium phase, where the goal of the transient phase is to steer the tumor to a desired composition, and the goal of the equilibrium phase is to maintain it at that composition.

In this section,
we have referred to a strategy alternating arbitrarily quickly between the maximum dose and no dose as the optimal long-run strategy.
Since this is an idealized schedule, it is more precise to say that 
there exists a sequence of pulsed strategies which achieve the optimal long-run growth rate
as $t_{\rm cycle}$ approaches zero.
However, same as for the case of linear induction, the cycle length can be extended significantly without affecting the long-run average growth rate of the tumor (Section \ref{app:pulsed}).    

\subsection{Integration with {\em in vitro} patient tumor data produces individualized clinical insights}
\label{sec:russo}

Linear and uniform induction of tolerance represent two extreme cases where the level of tolerance induction either increases gradually with the dose or
is insensitive to the dose.
In the former case, lower doses are generally preferred, whereas in the latter case, the maximum dose should be applied whenever the drug is present.
If a clinician understands which case applies better to a particular cancer type,
for example based on experience with prior patients,
these general insights can help guide clinical decision-making.
If the biological mechanism of drug-induced tolerance is furthermore known, it becomes possible to distinguish between Cases I, II and III.
For linear induction, this can help decide whether to apply low or high doses during the initial stages of treatment, and for uniform induction, it can help decide whether 
to apply the maximum dose continuously or intermittently in the long run.

Our approach 
can yield even more specific clinical insights
when integrated with experiments
involving {\em in vitro} models of patient tumors.
Technologies to derive clinically relevant cell culture models from patient tumor samples
are advancing rapidly \cite{wensink2021patient,jiang2020automated, tebon2023drug, walsh2014quantitative,gunnarsson2024understanding}, 
and their development has recently been put into focus by the FDA Modernization Act 2.0,
which seeks to replace animal models in drug safety and effectiveness testing by 
{\em in vitro} human cell models and {\em in silico} computational modeling \cite{fda_modernization_act_2_0,zushin2023fda,ahmed2023fda}.
As a case study,
we apply our insights to {\em in vitro} experiments from Russo et al.~\cite{russo2022modified}, where the authors treated two colorectal cancer cell lines with increasing doses of anti-cancer agents.
They found that when BRAF V600-E mutated WiDr colorectal cancer cells were treated with a combination of dabrafenib and cetuximab, the drugs linearly induced transitions from a drug-sensitive to a drug-persistence state (Case I under linear induction).
For RAS/RAF wild-type DiFi cells treated with cetuximab, the drug uniformly induced the same transitions (Case I under uniform induction).
We assume that the drug does not affect transitions from persistence to sensitivity, since no such effect is indicated in \cite{russo2022modified}.

\begin{figure}
    \centering
    \includegraphics[scale=1]{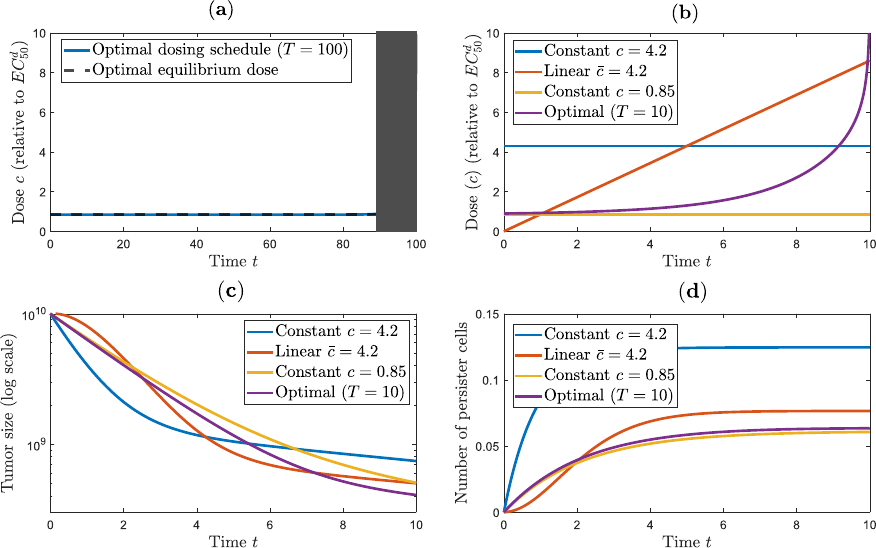}
    \caption{{\bf Optimal dosing for an {\em in vitro} colorectal cancer system.}
    {\bf (a)} 
    For BRAF V600-E mutated WiDr cells, Russo et al.~\cite{russo2022modified} observed linearly increasing transitions from sensitivity to tolerance
    under
    dabrafenib + cetuximab treatment.
    Under a sufficiently long time horizon, the optimal dosing strategy applies a low constant dose equal to 85\% of the ${\rm EC}_{50}^d$ dose.
    {\bf (b)} In \cite{russo2022modified}, the authors consider both a constant strategy and a linearly increasing strategy with average dose $\bar{c} = 4.2$ (relative to the ${\rm EC}_{50}^d$ dose) using a time horizon of $T=10$ days.
    The optimal strategy over $T=10$ days applies small doses during most of the period and increases the dose towards the end of the horizon.
    {\bf (c)} The linearly increasing schedule   
    with average dose $\bar{c} = 4.2$ outperforms the constant-dose schedule with the same average dose in terms of tumor reduction over $T=10$ days, which is consistent with the results of \cite{russo2022modified}.
    However, 
    the optimal strategy over $T=10$ days achieves a $19\%$ greater tumor reduction while applying half as much dose.
    Moreover, the optimal long-run constant dose $c^\ast=0.85$ achieves similar tumor reduction to the linear strategy at less than 20\% of the cumulative dose.
    {\bf (d)} Out of the four schedules considered in part (b), the optimal long-run constant dose  $c^\ast=0.85$ induces the fewest sensitive cells to adopt persistence over $T=10$ days.
    }
    \label{fig:russo}
\end{figure}

According to our analysis in  Section \ref{sec:uniforminduction}, 
we already know that for the DiFi cells, it is optimal to apply the maximum dose continuously.
We therefore focus 
on
the WiDr cells
and
compare the optimal strategy with dosing regimens considered in \cite{russo2022modified}. 
For the purposes of this analysis, we adopt the model of drug effect on cell proliferation from \cite{russo2022modified}
(Section \ref{app:russo}). 
As expected from our analysis in Section \ref{sec:optlinearinduction}, the optimal strategy over a treatment horizon of $T=100$ days applies a constant dose continuously, once again ignoring the boundary effect at the very end
(Figure \ref{fig:russo}(a)).
The optimal dose is even smaller here than for the examples considered in Section \ref{sec:optlinearinduction} or $c^\ast = 0.85$ as a proportion of the ${\rm EC}_{50}^d$ dose 
(in absolute terms, the ${\rm EC}_{50}^d$ dose is $2.3779 \cdot 10^{-7}$ $M$ and the optimal dose is $2 \cdot 10^{-7}$ $M$).
One reason that the optimal dose is lower than the ${\rm EC}_{50}^d$ dose in this case
is that
the dose at which the net sensitive cell division rate $\lambda_0(c)$ becomes negative is much smaller than the ${\rm EC}_{50}^d$ dose, meaning that the tumor reduces in size even at $c^\ast = 0.85$ (Section \ref{app:russo}).

In \cite{russo2022modified}, the authors considered a shorter time horizon of $T=10$ days, and they compared a constant dose schedule with $c=4.2$, relative to the ${\rm EC}_{50}^d$ dose,
with a linearly increasing dose schedule involving the same average dose $\bar{c} = 4.2$ over the horizon (Figure \ref{fig:russo}(b)).
They observed that the linearly increasing schedule significantly outperformed the constant-dose schedule, which is confirmed by our analysis (Figure \ref{fig:russo}(d)).
However, we note that the considerable preference for the linear schedule is due to the
constant-dose schedule with $c=4.2$ being significantly suboptimal.
For example, 
the long-run optimal dose
$c^\ast = 0.85$ leads to a similar tumor decay over $T=10$ days as the linear schedule with $\bar{c} = 4.2$ (Figure \ref{fig:russo}(c)).
The optimal short-term strategy over $T=10$ days
applies small doses for most of the period, starting at $c = 0.90$, 
before increasing
the dose at the very end (Figure \ref{fig:russo}(b)).
The optimal short-term strategy achieves a 19\% greater tumor reduction than the linear strategy, while applying less than half as much cumulative dose, or $\bar{c} = 2.0$ on average.

Here, we have evaluated the treatment schedules in terms of the tumor size at the end of the treatment horizon.
In \cite{russo2022modified}, the authors considered a different metric of treatment success, 
the number of persister cells induced under treatment, assuming no persister cells at the beginning of treatment.
Even under this metric, the optimal short-term strategy outperforms the linearly increasing schedule, 
and in fact,
the optimal long-run dose $c^\ast = 0.85$ performs the best (Figure \ref{fig:russo}(d)).
From one perspective, this is not surprising, since if the goal is to minimize the number of sensitive
cells that the drug induces to adopt persistence,
it is optimal to give no drug and thus induce 
no sensitive cells to adopt persistence.
On the other hand, 
the fact that the optimal long-run dose induces fewer persister cells than the other three schedules in the short term is precisely why it is better in the long run.

In summary, the constant-dose strategy with $c=0.85$ is both successful at restricting persister evolution in the short term and optimally controlling the tumor in the long term.
In addition, it applies around an 80\% lower cumulative dose than the linear schedule with average dose $\bar{c} = 4.2$.
These results show how applying our approach to experiments involving {\em in vitro} models of patient tumors can yield individualized clinical insights.
Determining the optimal dose, understanding how it affects the tumor in the short and long term, and comparing it with other potential dosing strategies would be difficult without the mathematical model.

\section{Discussion}

In this work, we have investigated optimal dosing of anti-cancer treatment under drug-induced plasticity.
We found that 
the optimal strategy always involves both a transient phase and an equilibrium phase,
where during the equilibrium phase, the tumor is maintained at a fixed 
intratumor composition that minimizes its long-run growth rate.
Under linear induction of tolerance,
the optimal equilibrium strategy is to apply a low constant dose, 
while under uniform induction of tolerance,
it is optimal to 
apply 
the maximum dose either continuously or intermittently.
We proved this mathematically and provided simple methods for computing the optimal strategies.
We also showed that during the transient phase, the optimal strategy steers the tumor to the desired long-run composition while precisely balancing the trade-off between cell kill and tolerance induction.
We discussed how these general insights can help guide clinical decision-making, and we showed how more specific insights can be obtained by integrating our methodology with experiments involving {\em in vitro} models of patient tumors.

We view this work as an initial step towards a more 
complete understanding of optimal dosing under drug-induced plasticity.
We have left many important questions unaddressed which serve as a guide for future work.
First, we have only considered the extreme cases where drug tolerance is induced in a linear or uniform fashion.
If the transition rate functions $\mu(c)$ and $\nu(c)$ follow more general Michaelis-Menten or Hill dynamics \cite{goutelle2008hill}, it may be difficult to apply optimal control analysis, but the optimal pulsed strategy can always be determined using the explicit expression \eqref{eq:sigmapulsed}.
Second, while we have assumed that cells transition directly between a drug-sensitive and a drug-tolerant state, these transitions may occur through intermediate states or in a more continuous fashion.
In addition, there is evidence that prolonged drug exposure can induce epigenetic reprogramming of tolerant cells and eventually lead to stable drug resistance \cite{sharma2010chromatin,shaffer2017rare}.
Tolerant cells can furthermore form a reservoir for
the evolution of eventual genetic resistance \cite{ramirez2016diverse,Engelman2016,russo2022modified}.
This suggests the need to include more cell states in the model, which is easily accommodated within our framework and we plan to do in future work.
Alternatively, the cell state can be treated as a continuous variable, representing gradual evolution towards resistance
\cite{chisholm2015emergence,chisholm2016cell,pouchol2018asymptotic,almeida2019evolution,jarrett2020optimal}.
Nevertheless, we believe that our analysis of the dynamics between sensitive and tolerant cells yields valuable insights that will be useful for understanding the more complex dynamics.
Third, we have considered an exponential growth model without competitive or cooperative dynamics.
If tumor growth is constrained by a carrying capacity,
our insights will continue to hold whenever it is possible to kill the tumor, 
since then the carrying capacity will not significantly affect the dynamics.
However, if recurrence cannot be avoided, it may become optimal to keep the tumor close to the carrying capacity, in order to avoid releasing tolerant cells from competition with sensitive cells \cite{gatenby2009change,greene2019mathematical,greene2020mathematical}.
This will be important to explore in future work.
Finally, we have equated drug dose with drug concentration throughout and ignored the effect of pharmacokinetics as well as unwanted toxic side effects of the drug.
Confirming the clinical applicability of our approach would require an investigation of these dynamics, particularly for the case of uniform tolerance induction, where the dynamics are significantly different depending on whether the drug is present or not.

Taking full advantage of our approach 
for a specific cancer type
or a specific patient
requires understanding
the quantitative dynamics 
of drug induction 
for that specific case.
Russo et al.~\cite{russo2022modified} have suggested a Bayesian approach for inferring the transition rate function $\mu(c)$ from sensitivity to tolerance using {\em in vitro} tumor bulk data, but their work only distinguishes between the extreme cases of linear and uniform induction of tolerance.
In a recent work, we have suggested a maximum likelihood framework for inferring more general dose-response dynamics for $\mu(c)$ based on the Hill equation \cite{wu2024inferring}.
However, 
as we have discussed,
it is also plausible that the drug inhibits the transition rate $\nu(c)$ from tolerance back to sensitivity, or that the drug simultaneously affects both rates.
In general, the present work shows that the optimal strategy varies significantly depending on the exact dynamics of tolerance induction.
This suggests the need 
to develop
integrated experimental and mathematical tools capable of jointly inferring $d_0(c)$, $\mu(c)$ and $\nu(c)$ from experimental data,
which may require data on the phenotypic composition of the tumor \cite{gunnarsson2023statistical}.
Ultimately, 
these integrated tools can help usher in a new era of precision medicine, 
where the dynamics of drug induction are determined
and dosing strategies are optimized
on an individual patient basis \cite{mcguire2013formalizing,hamis2019blackboard}.

\section{Models and methods}

\subsection{Model of cell proliferation and phenotypic switching} \label{sec:model}

Our baseline model is a multi-type branching process model in continuous time \cite{athreya2004branching}.
In the model, cells switch stochastically between two distinct cell states, a drug-sensitive (type-0) and a drug-tolerant (type-1) state.
In the absence of drug, a type-0 cell divides into two type-0 cells at rate $b_0$, it dies at rate $d_0$ and it transitions to the tolerant type-1 state at rate $\mu>0$.
More precisely, each type-0 cell waits an exponentially distributed amount of time with rate $a_0 := b_0+d_0+\mu$ 
before it either divides with probability $b_0/a_0$, dies with probability $d_0/a_0$, or transitions to type-1 with probability $\mu/a_0$.
A type-1 cell divides at rate $b_1$, it dies at rate $d_1$, and it reverts to type-0 at rate $\nu>0$ (Figure \ref{fig:model_overviews}(a)).
The net birth rates of the two types are denoted by $\lambda_0 := b_1-d_1$ and $\lambda_1:=b_1-d_1$.

\subsection{Drug effect on cell proliferation} \label{sec:drugeffectprolif}

To model the effect of the anti-cancer drug on cell proliferation, we assume that the sensitive cell death rate $d_0$ is an increasing function of the current drug dose $c$.
For simplicity, we assume that the drug-tolerant cells are unaffected by the drug.
The specific functional form for $d_0$ is the Michaelis-Menten function, which is commonly used for this purpose:
\begin{align} \label{eq:d0fcn}
    &d_0(c) = d_0+(d_{\rm max}-d_0)\cdot\frac{c}{c+{\rm EC}_{50}^d}
    = d_0+\Delta d_0 \cdot\frac{c}{c+{\rm EC}_{50}^d}.
\end{align}
Here, $d_0$ is the sensitive cell death rate in the absence of drug, $d_{\rm max}$ is the saturation death rate under an arbitrarily large drug dose, $\Delta d_0 := d_{\rm max}-d_0 \geq 0$ is their difference, and ${\rm EC}_{50}^d$ is the dose at which the drug has half the maximal effect on sensitive cell proliferation.
Note that we write simply $d_0$ instead of $d_0(0)$ to signify the death rate in the absence of drug.
The function in \eqref{eq:d0fcn} is concave, but it is worth noting that if the 
the drug dose is viewed on a log scale, then \eqref{eq:d0fcn} becomes an $S$-shaped dose response curve which is characteristic of viability curves in the biological literature (Figure \ref{fig:logdose}).

Expression \eqref{eq:d0fcn} can be rewritten as
\begin{align*}
    &d_0(c) = d_0+\Delta d_0\cdot\frac{(c/{\rm EC}_{50}^d)}{(c/{\rm EC}_{50}^d)+1}.
\end{align*}
Thus, if we measure the drug dose 
as a proportion of 
the ${\rm EC}_{50}^d$ dose, we obtain the simpler expression
\begin{align} \label{eq:d0fcnsimpler}
    &d_0(c) = d_0+\Delta d_0\cdot\frac{c}{c+1}.
\end{align}
We will usually adopt this scaling of the drug dose and use the simpler form for $d_0(c)$.
Similarly, the net birth rate of sensitive cells is given by
\begin{align} \label{eq:lambda0}
    \lambda_0(c) = b_0-d_0(c) = \lambda_0 - \Delta d_0 \cdot\frac{c}{c+1},
\end{align}
where we simply write $\lambda_0$ instead of $\lambda_0(0)$.

While we have assumed in the preceding discussion that the drug increases the death rate of sensitive cells (cytotoxic drug), we can equivalently assume that the drug decreases their division rate (cytostatic drug).
This is because our results depend on the division rate and death rate of sensitive cells only through the net birth rate $\lambda_0(c)$.
If we assume that a cytostatic drug influences the division rate as follows:
\begin{align*}
    b_0(c) = b_0 - \Delta b_0 \cdot \frac{c}{c+1},
\end{align*}
the net division rate becomes
\begin{align*}
    \lambda_0(c) = b_0(c) - d_0 = \lambda_0 - \Delta b_0 \cdot \frac{c}{c+1},
\end{align*}
which has the same form as expression \eqref{eq:lambda0}.

\begin{figure}
    \centering
    \includegraphics[scale=0.9]{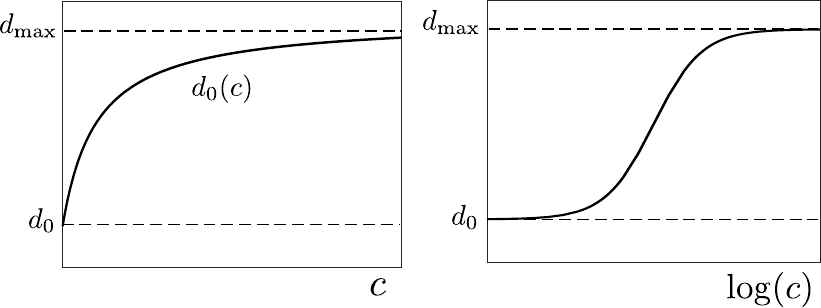}
    \caption{
    {\bf Drug effect on sensitive cell death rate $d_0$.}
    The Michaelis-Menten function \eqref{eq:d0fcn} describing the drug effect on cell proliferation is concave in the drug dose $c$ (left panel).
    When the dose $c$ is viewed on a log scale, the function becomes takes on an $S$-shape which is characteristic of viability curves in the biological literature (right panel).
    }
    \label{fig:logdose}
\end{figure}

\subsection{Drug effect on switching rates} \label{sec:drugeffectswitching}

To model the possibility of drug-induced tolerance, we assume that in the presence of the anti-cancer drug, one or both of the transition rates $\mu$ and $\nu$ change with the drug dose $c$.
If we focus first on the transition rate $\mu$ from sensitivity to tolerance, it is natural to assume that it follows a Michaelis-Menten function of the form \eqref{eq:d0fcn}:
\begin{align} \label{eq:mufcn}
    \mu(c) &= \mu_0+(\mu_{\rm max}-\mu_0) \cdot \frac{c}{c+{\rm EC}_{50}^{\mu}} 
    = \mu_0+\Delta \mu \cdot \frac{c}{c+{\rm EC}_{50}^{\mu}},
\end{align}
where $\Delta \mu := \mu_{\rm max}-\mu_0 \geq 0$.
The dose at which the drug has half the maximal effect on $\mu$, ${\rm EC}_{50}^\mu$, is in general distinct from ${\rm EC}_{50}^d$.
Here, we will focus on the two extreme cases where ${\rm EC}_{50}^{\mu} \gg {\rm EC}_{50}^{d}$ and ${\rm EC}_{50}^{\mu} \ll {\rm EC}_{50}^{d}$.
The case ${\rm EC}_{50}^{\mu} \gg {\rm EC}_{50}^{d}$ means that the drug induces tolerance at doses which are already very effective at killing the sensitive cells.
In this case, if the drug dose $c$ is similar in magnitude to ${\rm EC}_{50}^{d}$,
then $c \ll {\rm EC}_{50}^\mu$ and by Taylor expansion:
\begin{align*}
    \mu(c) =  \mu_0+\Delta \mu \cdot \frac{{c}/{\rm EC}_{50}^{\mu}}{{c}/{{\rm EC}_{50}^{\mu}}+1} \approx \mu_0 + \frac{\Delta \mu}{{\rm EC}_{50}^\mu} \cdot c.
\end{align*}
In other words, $\mu(c)$ can be treated as linearly increasing in $c$.
The case ${\rm EC}_{50}^{\mu} \ll {\rm EC}_{50}^{d}$ means that the drug has fully induced tolerance at low doses which are not yet effective at killing the sensitive cells.
In this case, if the drug dose $c$ is similar in magnitude to ${\rm EC}_{50}^{d}$, then $c \gg {\rm EC}_{50}^\mu$ and by Taylor expansion 
for $c>0$:
\begin{align*}
    \mu(c) =\mu_0+\Delta \mu \cdot \frac{1}{1+{\rm EC}_{50}^{\mu}/c} \approx \mu_0 + \Delta \mu = \mu_{\rm max}.
\end{align*}
In other words, $\mu$ is effectively constant in the presence of drug.

Based on the preceding, we assume the following two functional forms for $\mu$:
\begin{enumerate}[(1)]
    \item $\mu(c) = \mu_0 + k c$, where $k \geq 0$ (tolerance is linearly induced).
    \item $\mu(0) = \mu_0$ and $\mu(c) = \mu_{\rm max}$ for $c>0$ (tolerance is uniformly induced).
\end{enumerate}
The reason we focus on these two extremes is twofold. 
First, these two exact forms were observed empirically in the recent investigation by Russo et al.~\cite{russo2022modified}, which is the only work we know of which has attempted to quantify drug-induced tolerance in a dose-dependent manner using {\em in vitro} data.
Second, due to their simplicity relative to the general case \eqref{eq:mufcn}, we are able to establish rigorous mathematical results on optimal dosing strategies.

For the transition rate $\nu$ from tolerance back to sensitivity, the corresponding Michaelis-Menten function has the following form:
\begin{align*} 
    \nu(c) &= \nu_0+(\nu_{\rm min}-\nu_0) \cdot \frac{c}{c+{\rm EC}_{50}^{\nu}} 
    = \nu_0 - \Delta \nu \cdot \frac{c}{c+{\rm EC}_{50}^{\nu}},
\end{align*}
where $\Delta \nu:= \nu_0 - \nu_{\rm min} \geq 0$.
We will also assume the following two functional forms for
$\nu$:
\begin{enumerate}[(1)]
    \item $\nu(c) = \nu_0 - m c$, where $m \geq 0$.
    \item $\nu(0) = \nu_0$ and $\nu(c) = \nu_{\rm min}$ for $c>0$.
\end{enumerate}
By assumption, $\nu$ is a decreasing function of $c$, since we are modeling the case where the drug inhibits transitions from drug-tolerance to drug-sensitivity.

Throughout, we assume that $\mu_0>0$ and $\nu_0>0$, meaning that cells are able to transition between the two states even in the absence of drug.

\subsection{System equations} \label{sec:systemeqns}

Let $Z_0(t)$ and $Z_1(t)$ denote the number of sensitive and tolerant cells, respectively, at time $t$.
Let $n_0(t) := \E[Z_0(t)]$ and $n_1(t) := \E[Z_1(t)]$ be the mean number of sensitive and tolerant cells at time $t$.
These mean functions satisfy the differential equations
\begin{align} \label{eq:systemeqn}
\begin{split}
        & \frac{dn_0}{dt} = (\lambda_0(c)-\mu(c))n_0+\nu(c)n_1, \\
    & \frac{dn_1}{dt} = (\lambda_1-\nu(c))n_1+\mu(c)n_0.
\end{split}
\end{align}
Here, the dose $c$ should be viewed as a function of time, $c(t)$.
The total population size $n(t) := n_0(t)+n_1(t)$ obeys the differential equation
\begin{align} \label{eq:dndt0}
    \frac{dn}{dt} = \lambda_0(c)n_0+\lambda_1 n_1.
\end{align}
Let $f_0(t) := n_0(t)/n(t)$ and $f_1(t) := 1-f_0(t)$ be the proportions of each cell type in the population at time $t$.
Then
\begin{align} \label{eq:df0dt}
    \frac{df_0}{dt} &= \frac{({dn_0}/{dt}) \cdot n - n_0 \cdot ({dn}/{dt})}{n^2} \nonumber \\
    &= \frac{\big((\lambda_0(c)-\mu(c))n_0+\nu(c)n_1\big) n - n_0 \big(\lambda_0(c)n_0+\lambda_1n_1\big)}{n^2}  \nonumber\\
    &= (\lambda_0(c)-\mu(c))f_0+\nu(c)(1-f_0) - f_0 \big(\lambda_0(c)f_0+\lambda_1(1-f_0)\big) \nonumber \\
    &= \big(\lambda_1-\lambda_0(c)\big)f_0^2 - \big(\lambda_1-\lambda_0(c)+\mu(c)+\nu(c)\big)f_0+\nu(c). 
\end{align}
Expression \eqref{eq:dndt0} can be rewritten as follows:
\begin{align} \label{eq:dndt}
        \frac{dn}{dt} &= \big(\lambda_0(c)f_0+\lambda_1(1-f_0)\big) n = \big((\lambda_0(c)-\lambda_1)f_0+\lambda_1\big)n.
\end{align}
Solving this differential equation,
the mean tumor size at time $T$ is given by the explicit expression
\begin{align} \label{eq:n(t)formula}
    n(T) = n(0) \exp\left(\int_0^T \big((\lambda_0(c)-\lambda_1)f_0+\lambda_1\big)dt\right).
\end{align}
Expressions \eqref{eq:dndt} and \eqref{eq:n(t)formula} indicate that the instantaneous growth rate of the population at time $t$ is given by
\begin{align} \label{eq:growthratedef}
    u(c,f_0) := (\lambda_0(c)-\lambda_1)f_0+\lambda_1,
\end{align}
which will be an important quantity in our investigation of optimal treatment strategies.

\subsection{Behavior under a constant dose} \label{sec:constantdoseequilibrium}

When a constant drug dose is applied, $c(t) = c$, all model parameters $\lambda_0(c)$, $\lambda_1$, $\mu(c)$ and $\nu(c)$ are constant.
In this case, the system equations \eqref{eq:systemeqn} 
admit explicit solutions.
To simplify the notation, we let $\lambda_0$, $\lambda_1$, $\mu$ and $\nu$ denote the constant parameters.
We first note that the dynamics of the system can be encoded in the so-called {\em infinitesimal generator}
\begin{align*}
    {\bf A} := \begin{bmatrix}
        \lambda_0-\mu & \mu \\ \nu & \lambda_1-\nu,
    \end{bmatrix},
\end{align*}
where the $(i+1,j+1)$-th element describes the net rate at which a cell of type-$i$ produces a cell of type-$j$.
The infinitesimal generator has distinct real eigenvalues $\rho < \sigma$ given by
\begin{align*}
    & \sigma = \frac12 \left((\lambda_0-\mu) + (\lambda_1-\nu) + \sqrt{((\lambda_0-\mu)-(\lambda_1-\nu))^2+4\mu\nu}\right), \\
    & \rho = \frac12 \left((\lambda_0-\mu) + (\lambda_1-\nu) - \sqrt{((\lambda_0-\mu)-(\lambda_1-\nu))^2+4\mu\nu}\right).
\end{align*}
Now define
\begin{align*}
    & \delta := \frac{(\lambda_0-\mu)-\rho}{\nu}, \\
    & \beta := \frac{\sigma-(\lambda_0-\mu)}{\nu}.
\end{align*}
If $n_0(0)=n$ and $n_1(0)=m$, the mean number of cells of each type can be written explicitly as
\begin{align} \label{eq:n0n1expl}
\begin{split}
        & n_0(t) = \frac{n\delta+m}{\delta+\beta} e^{\sigma t} + \frac{n\beta-m}{\delta+\beta} e^{\rho t}, \\
    & n_1(t) = \frac{\beta(n\delta+m)}{\delta+\beta} e^{\sigma t} - \frac{\delta(n\beta-m)}{\delta+\beta} e^{\rho t}.
\end{split}
\end{align}
These expressions
indicate
two important aspects of the long-run dynamics.
First, the tumor eventually grows (or decays) at 
exponential rate $\sigma$.
Second, the intratumor composition eventually reaches
an equilibrium where the proportion between type-1 and type-0 becomes the constant $\beta$.
We denote the equilibrium proportion of the type-0 cells as
\begin{align} \label{eq:f0bar}
    \bar{f}_0 := \frac{1}{1+\beta} = \frac{\nu}{\sigma-\lambda_0+\mu+\nu}.
\end{align} 
When we are considering a constant treatment, $c(t) = c$, we will write $\sigma(c)$ and $\bar{f}_0(c)$ to explicitly denote the dependence on the dose $c$.
For detailed derivations of the above formulas, we refer to Appendix A.2 of our previous work \cite{gunnarsson2020understanding}.

\subsection{Optimal control problem}
\label{sec:optcontrolproblem}

To determine the optimal dosing strategy, we assume a finite treatment horizon $[0,T]$ with $T>0$.
Let ${\rm PC}([0,T])$ denote the set of all piecewise continuous functions on $[0,T]$.
The space of allowable dosing strategies is
\[
{\cal C}_T := \{c \in {\rm PC}([0,T]): c(t) \leq c_{\rm max}, \, \forall t \in [0,T] \},
\]
where $c_{\rm max}$ is the maximum allowable instantaneous dose.
For simplicity, when deriving optimal dosing strategies, we assume that the drug dose corresponds perfectly to the drug concentration reaching the tumor, thereby ignoring pharmacokinetic effects.
Our main objective is to find the dosing strategy $c \in {\cal C}_T$ which minimizes the expected tumor size at the end of the treatment period.
More specifically, we aim to solve the optimal control problem
\begin{align*}
    \inf_{c \in {\cal C}_T} n(T) = \inf_{c \in {\cal C}_T} (n_0(T)+n_1(T)). 
\end{align*}
By \eqref{eq:n(t)formula} and \eqref{eq:growthratedef}, it is equivalent to solve the problem
\begin{align} \label{eq:optcontrolproblem}
    \inf_{c \in {\cal C}_T} \int_0^T u(c,f_0) dt,
\end{align}
where $u(c,f_0) := (\lambda_0(c)-\lambda_1)f_0+\lambda_1$ is the instantaneous growth rate of the population and $f_0(t)$ is governed by the differential equation 
\begin{align*} 
    \frac{df_0}{dt} &= \big(\lambda_1-\lambda_0(c)\big)f_0^2 - \big(\lambda_1-\lambda_0(c)+\mu(c)+\nu(c)\big)f_0+\nu(c). 
\end{align*}
Note that under the assumption that $c$ is piecewise continuous on $[0,T]$, $f_0$ is piecewise continuously differentiable on $[0,T]$.
Over a finite treatment horizon $[0,T]$, the average growth rate of the population is given by
\begin{align*}
    \bar{u}_T(c) = \frac1T \int_0^T u(c,f_0)dt.
\end{align*}
In our investigation of the optimal control problem \eqref{eq:optcontrolproblem}, we will be interested in knowing which treatment is the best in the long run as $T \to \infty$.
Since the integral in \eqref{eq:optcontrolproblem} can be unbounded as $T \to \infty$, we will also consider the {\em long-run average growth rate}
\begin{align} \label{eq:longrunavgrate}
    \bar{u}_\infty(c) = \limsup_{T \to \infty} \bar{u}_T(c) = \limsup_{T \to \infty} \frac1T \int_0^T u(c,f_0)dt,
\end{align}
where $c: [0,\infty) \to [0,c_{\rm max}]$ is considered a treatment with an infinite time horizon.

\subsection{Rapid pulsed strategies} \label{sec:rapidpulsed}

Pulsed schedules are simple treatment strategies which alternate
between giving a fixed drug dose and no drug.
We identify a pulsed schedule as a three-dimensional vector $(c,t_{\rm cycle},\varphi)$ where $c$ is the dose applied, $t_{\rm cycle}$ is the duration of each treatment cycle, and $\varphi$ is the proportion of drug application
during
each cycle.
We will be particularly interested  in idealized strategies where the treatment cycles are arbitrarily short.
We note that if the cycle time $t_{\rm cycle} = dt$ is infinitesimal,
then using \eqref{eq:df0dt},
the infinitesimal change in $f_0$ is
\begin{align*} 
    df_0 
    &= \varphi dt \big(\big(\lambda_1-\lambda_0(c)\big)f_0^2 - \big(\lambda_1-\lambda_0(c)+\mu(c)+\nu(c)\big)f_0+\nu(c)\big) \\
    & + (1-\varphi)dt \big(\big(\lambda_1-\lambda_0\big)f_0^2 - \big(\lambda_1-\lambda_0+\mu_0+\nu_0\big)f_0+\nu_0\big).
\end{align*}
If we define the average rates over each cycle,
\begin{align} \label{eq:avgratescycle}
\begin{split}
        & \bar\lambda_0(c,\varphi) := \varphi \lambda_0(c) + (1-\varphi)\lambda_0 = \lambda_0 + \varphi (\lambda_0(c)-\lambda_0), \\
    &  \bar\mu(c,\varphi) := \varphi \mu(c) + (1-\varphi)\mu_0 = \mu_0 + \varphi (\mu(c)-\mu_0), \\
    &  \bar\nu(c,\varphi) := \varphi \nu(c) + (1-\varphi)\nu_0 = \nu_0 + \varphi (\nu(c)-\nu_0),
\end{split}
\end{align}
then $f_0(t)$ obeys the differential equation
\begin{align} \label{eq:df0ftrapidpulse}
    \frac{df_0}{dt} &= \big(\lambda_1-\bar\lambda_0(c,\varphi)\big)f_0^2 - \big(\lambda_1-\bar\lambda_0(c,\varphi)+\bar\mu(c,\varphi)+\bar\nu(c,\varphi)\big)f_0+\bar\nu(c,\varphi).
\end{align}
Similarly, the instantaneous growth rate of the population is given by
\begin{align*}
    \bar{u}(c,\varphi,f_0) = (\bar\lambda_0(c,\varphi)-\lambda_1)f_0 + \lambda_1.
\end{align*}
Thus, under an arbitrarily fast pulsed schedule, the associated model dynamics can be approximated by a constant-dose model with parameters $\bar{\lambda}_0(c,\varphi)$, $\lambda_1$, $\bar\mu(c,\varphi)$ and $\bar\nu(c,\varphi)$.
By Section \ref{sec:constantdoseequilibrium}, in the long run, the population grows at exponential rate
\begin{align} \label{eq:sigmapulsed}
    \sigma(c,\varphi) :=  \frac12 \Big(&(\bar\lambda_0(c,\varphi)-\bar\mu(c,\varphi)) + (\lambda_1-\bar\nu(c,\varphi)) \nonumber \\
    &+ \sqrt{((\bar\lambda_0(c,\varphi)-\bar\mu(c,\varphi))-(\lambda_1-\bar\nu(c,\varphi)))^2+4\bar\mu(c,\varphi)\bar\nu(c,\varphi)}\Big).
\end{align}
In addition, the population eventually reaches an equilibrium composition, where the proportion of sensitive cells becomes
\begin{align} \label{eq:f0barpulseed}
    \bar{f}_0(c,\varphi) := \frac{\bar\nu(c,\varphi)}{\sigma(c,\varphi)-\bar\lambda_0(c,\varphi)+\bar\mu(c,\varphi)+\bar\nu(c,\varphi)}.
\end{align}

\subsection{Baseline parameters} \label{sec:parametrization}

Our computational results are shown using a parametrization of the model inspired by experimental investigations of drug tolerance in cancer and bacteria. The parameter values are shown in Table \ref{table:parameters}.
First, we assume that in the absence of drug, the proliferation rate of sensitive cells is $\lambda_0 = 0.04$, and that the maximal drug effect on the death rate is $\Delta d_0 = 0.08$.
Our motivating example is PC9 cells treated with a large dose of erlotinib, for which the proliferation rate of sensitive cells in the absence of drug has been estimated as $\lambda_0 = 0.04$ per hour and the drug effect as $\Delta d_0 = 0.08$ per hour \cite{gunnarsson2020understanding}.
We note that for a cohort of patients with metastatic melanoma treated with vemurafenib, the typical value of $\lambda_0$ was 0.01 per day and the typical drug effect was $\Delta d_0 = 0.04$ per day.
Thus, the time unit for the baseline regime can be considered to be hours in the {\em in vitro} setting and days in the clinical setting.

Second, we assume that in the absence of drug, $\lambda_1 \ll \lambda_0$, meaning that the drug-tolerant cells proliferate much slower than the drug-sensitive cells, and that they only make up a small proportion of the tumor at the start of the treatment.
This is consistent with experimental evidence
showing that drug-tolerant cells are generally slow-cycling
\cite{sharma2010chromatin,roesch2010temporarily,pisco2013non,roesch2013overcoming,Engelman2016,shaffer2017rare}.
Third, we assume that $\mu_0 \ll \lambda_0$, meaning that transitions from sensitivity to tolerance are rare compared to cell divisions.
This is consistent with evidence that epigenetic modifications are commonly retained for $10 - 10^5$ cell divisions \cite{niepel2009non, sigal2006variability, cohen2008dynamic}.
Fourth, we assume that $\nu_0$ is significantly larger than $\mu_0$, meaning that the tolerant state is lost faster than it is adopted.
This is consistent with experimental evidence both from cancer \cite{gupta2011stochastic,paryad2021optimal} and bacteria \cite{norman2015stochastic}.

In Figures \ref{fig:optimallinear} and \ref{fig:pulsedlinear}, we assume that $k=m=0.0004$ and that the maximum allowable dose is $c_{\rm max}=10$ as a proportion of the ${\rm EC}_{50}^d$ dose.
In 
Figure \ref{fig:optimalheaviside}, we assume that $\Delta \mu = 0.004$, $\Delta \nu = 0.003$ and $c_{\rm max}=10$.
Unless otherwise noted, we assume that the 
proportion of sensitive cells at the start of treatment is the equilibrium proportion $\bar{f}_0(0)$ in the absence of drug.

We note that with this parametrization, the net growth rate of tolerant cells taking phenotypic switching into account is $\lambda_1-\nu_0 = -0.003 < 0$.
Thus, in the presence of drug, the tolerant population decays over time, but much slower than the sensitive cells.

\begin{table}[h]
\centering
\begin{tabular}{|c|c|c|c|c|c|c|c|c|c|c|c|}
\hline
Parameter&$\lambda_0$ & $\Delta d_0$& $\lambda_1$& $\mu_0$& $\nu_0$&$k$&$m$&$\Delta \mu$&$\Delta \nu$&$c_{\rm max}$ \\
\hline
Value&0.04&0.08&0.001&0.0004&0.004&0.0004&0.0004&0.004&0.003&10 \\
\hline
\end{tabular}
\caption{
Baseline parameter values used in the study.
}
\label{table:parameters}
\end{table}

\noindent {\bf Data availability.} 
All data generated or analyzed during this study are included in this published article and its supplementary information files. \\

\noindent {\bf Code availability.} All codes necessary to reproduce the results of this study 
are available in the Github repository \url{https://github.com/egunnars/optimal_dosing_drug-induced_plasticity}.  \\

\noindent {\bf Acknowledgments.} The work of EBG was supported in part by NIH/NCI grant R01 CA241137.
The work of JF was supported in part by NSF grants CMMI-2228034 and DMS-2052465.  \\

\noindent {\bf Competing interests.} All authors declare no financial or non-financial competing interests.  \\

\bibliographystyle{naturemag}
\bibliography{epi}

\newpage

\appendix

\makeatletter
\renewcommand\thesection{S\@arabic\c@section}
\renewcommand\thefigure{\thesection.\arabic{figure}} 
\renewcommand\thetable{\thesection.\arabic{table}} 
\renewcommand\theequation{\thesection.\arabic{equation}} 

\setcounter{figure}{0}  
\setcounter{table}{0}  
\setcounter{equation}{0} 
\setcounter{page}{1}

\section{Supplementary text}

\subsection{Review of mathematical model and optimal control problem}

In the interest of self-containment, we briefly review the mathematical model and the optimal control problem studied in the main text.
In the model, cells are able to transition between two states, a drug-sensitive (type-0) and a drug-tolerant (type-1) state.
Type-0 cells divide at rate $b_0$, die at rate $d_0$ and transition to type-1 at rate $\mu$.
Type-1 cells divide at rate $b_1$, die at rate $d_1$ and transition to type-0 at rate $\nu$.
The net birth rates of the two cell types are $\lambda_0 := b_0-d_0$ and $\lambda_1 := b_1-d_1$.
To model the effect of the anti-cancer drug on sensitive cell proliferation, we assume that the death rate $d_0(c)$ increases with the drug dose $c$ (equivalently, it can be assumed that the drug affects the division rate $b_0$).
To model drug-induced tolerance, we assume that one or both of the transition rates $\mu$ and $\nu$ are dose-dependent.
Two specific functional forms for $\mu(c)$ and $\nu(c)$ are considered, modeling either linear or uniform induction of tolerance.
An overview of the model is given in Figure \ref{fig:model_overviews_suppl}.

\begin{figure}
    \centering
    \includegraphics[scale=0.9]{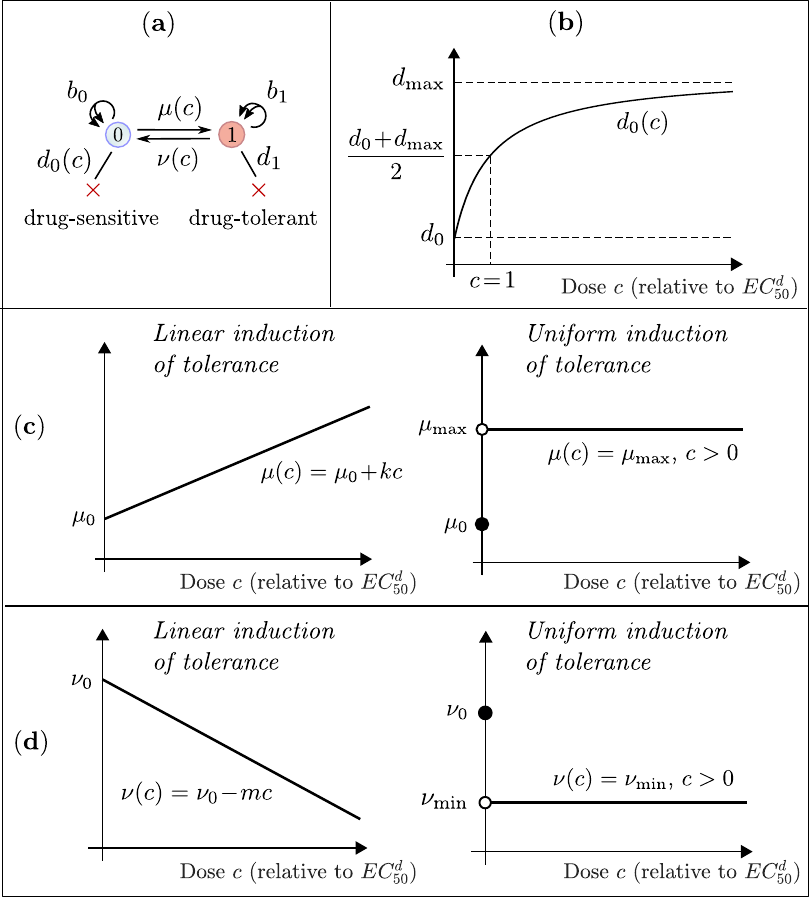}
      \caption{{\bf Mathematical model of drug-induced tolerance (reprint of Figure \ref{fig:model_overviews} from main text).}
    {\bf (a)} In the mathematical model, cells transition between two states, a drug-sensitive (type-0) state and a drug-tolerant (type-1) state.
    Sensitive cells divide at rate $b_0$, die at rate $d_0$ and transition to tolerance at rate $\mu$.
    Tolerant cells divide at rate $b_1$, die at rate $d_1$ and transition to sensitivity at rate $\nu$.
    The sensitive cell death rate $d_0$ and the transition rates $\mu$ and $\nu$ 
    depend on the drug dose $c$ (Section \ref{sec:model}).
    {\bf (b)} The sensitive cell death rate follows a Michaelis-Menten equation of the form $d_0(c) = d_0 + (d_{\rm max}-d_0) c/(c+1)$, where $d_0$ is the death rate in the absence of drug and 
    $d_{\rm max}$ is the saturation death rate under an arbitrarily large drug dose.
    The dose $c$ is normalized to the ${\rm EC}_{50}$ dose for $d_0$, meaning 
    that the drug has half the maximal effect at dose $c=1$ (Section \ref{sec:drugeffectprolif}).
    {\bf (c)} The transition rate $\mu(c)$ from drug-sensitivity to drug-tolerance is assumed either a linearly increasing function of  $c$ or to be uniformly elevated in the presence of drug.
    These two forms of drug-induced tolerance were observed in a  recent study by Russo et al.~\cite{russo2022modified} and they also emerge as limiting cases of more general Michaelis-Menten dynamics (Section \ref{sec:drugeffectswitching}).
    {\bf (d)} The transition rate $\nu(c)$ from drug-tolerance to drug-sensitivity is similarly assumed either a linearly decreasing function of $c$ or uniformly inhibited in the presence of drug (Section \ref{sec:drugeffectswitching}).
    }
        \label{fig:model_overviews_suppl}
\end{figure}

The model dynamics can be described by the differential equations
\begin{align*} 
        & \frac{dn_0}{dt} = (\lambda_0(c)-\mu(c))n_0+\nu(c)n_1, \\
    & \frac{dn_1}{dt} = (\lambda_1-\nu(c))n_1+\mu(c)n_0,
\end{align*}
where $n_i(t)$ denotes the {expected number} of cells of type-$i$,
and 
the dose is viewed as a function of time, $c(t)$.
The proportion $f_0(t) = n_0(t)/(n_0(t)+n_1(t))$ of sensitive cells at time $t$ is governed by the differential equation
\begin{align} \label{eq:df0ftsuppl}
    \frac{df_0}{dt} &= \big(\lambda_1-\lambda_0(c)\big)f_0^2 - \big(\lambda_1-\lambda_0(c)+\mu(c)+\nu(c)\big)f_0+\nu(c). 
\end{align}
Under a constant dose $c(t)=c$, the model dynamics simplify, and the tumor eventually grows exponentially at rate
\begin{align*}
     \sigma(c) = \frac12 \Big(&\lambda_0(c)-\mu(c) + \lambda_1-\nu(c) + \sqrt{(\lambda_0(c)-\mu(c)-\lambda_1+\nu(c))^2+4\mu(c)\nu(c)}\Big),
\end{align*}
which we refer to as the {\em equilibrium growth rate} under a constant dose $c$,
and the proportion of sensitive cells eventually reaches the equilibrium value
\begin{align*} 
    \bar{f}_0(c) = \frac{\nu(c)}{\sigma(c)-\lambda_0(c)+\mu(c)+\nu(c)}.
\end{align*}

Under a pulsed schedule $(c,\varphi)$  alternating arbitrarily fast between the dose $c$ and no dose, if we define the following average rates over each treatment cycle,
\begin{align} \label{eq:avgratessuppl}
\begin{split}
            & \bar\lambda_0(c,\varphi) := \varphi \lambda_0(c) + (1-\varphi)\lambda_0 = \lambda_0 + \varphi (\lambda_0(c)-\lambda_0), \\
    &  \bar\mu(c,\varphi) := \varphi \mu(c) + (1-\varphi)\mu_0 = \mu_0 + \varphi (\mu(c)-\mu_0), \\
    &  \bar\nu(c,\varphi) := \varphi \nu(c) + (1-\varphi)\nu_0 = \nu_0 + \varphi (\nu(c)-\nu_0),
\end{split}
\end{align}
then $f_0(t)$ obeys the differential equation
\begin{align} \label{eq:df0ftrapidpulsesuppl}
    \frac{df_0}{dt} &= \big(\lambda_1-\bar\lambda_0(c,\varphi)\big)f_0^2 - \big(\lambda_1-\bar\lambda_0(c,\varphi)+\bar\mu(c,\varphi)+\bar\nu(c,\varphi)\big)f_0+\bar\nu(c,\varphi).
\end{align}
In this case, the tumor eventually grows at exponential rate 
\begin{align} \label{eq:sigmapulsedsuppl}
    \sigma(c,\varphi) =  \frac12 \Big(&(\bar\lambda_0(c,\varphi)-\bar\mu(c,\varphi)) + (\lambda_1-\bar\nu(c,\varphi)) \nonumber \\
    &+ \sqrt{((\bar\lambda_0(c,\varphi)-\bar\mu(c,\varphi))-(\lambda_1-\bar\nu(c,\varphi)))^2+4\bar\mu(c,\varphi)\bar\nu(c,\varphi)}\Big),
\end{align}
and the proportion of sensitive cells eventually reaches the equilibrium value
\begin{align*} 
    \bar{f}_0(c,\varphi) = \frac{\bar\nu(c,\varphi)}{\sigma(c,\varphi)-\bar\lambda_0(c,\varphi)+\bar\mu(c,\varphi)+\bar\nu(c,\varphi)}.
\end{align*}

We consider a dosing strategy $c: [0,T] \to [0,c_{\rm max}]$ optimal if it minimizes the total tumor size, $n(T) = n_0(T)+n_1(T)$, at the end of a finite time-horizon $[0,T]$.
The total tumor size $n(t)$ at time $t$
follows the differential equation
\begin{align} \label{eq:dndtsuppl}
        \frac{dn}{dt} & = \big((\lambda_0(c)-\lambda_1)f_0+\lambda_1\big)n,
\end{align}
which implies that we can 
reframe our optimal control problem 
in terms of minimizing
\begin{align} \label{eq:optcontobjsuppl}
  \int_0^T u(c,f_0) dt,
\end{align}
where 
\begin{align} \label{eq:instgrowthratesuppl}
    u(c,f_0) := (\lambda_0(c)-\lambda_1)f_0+\lambda_1
\end{align}
is the instantaneous tumor growth rate.
In our analysis of optimal dosing strategies over a long time horizon,
we also consider the problem of minimizing the {\em long-run average tumor growth rate}
\begin{align} \label{eq:longrunavgratesuppl}
    \bar{u}_\infty(c) = \limsup_{T \to \infty} \frac1T \int_0^T u(c,f_0)dt,
\end{align}
where $c: [0,\infty) \to [0,c_{\rm max}]$ is considered an infinite-horizon treatment.

\subsection{Implementation of forward-backward sweep method} \label{app:fbsm}

For the case of linear induction of tolerance, we use the forward-backward sweep method \citesupp{mcasey2012convergence,sharp2022numerical} to solve the optimal control problem
\begin{align*}
    \min_{c} \int_0^T u(c,f_0) dt,
\end{align*}
where 
$f_0(t)$ is governed by the differential equation 
\begin{align*} 
    \frac{df_0}{dt} &= g(f_0,c)
\end{align*}
with
\[
    g(f_0,c) := \big(\lambda_1-\lambda_0(c)\big)f_0^2 - \big(\lambda_1-\lambda_0(c)+\mu(c)+\nu(c)\big)f_0+\nu(c). 
\]
The so-called {\em Hamiltonian} for this problem is
\begin{align} \label{eq:hamiltonian}
    {\cal H}(f_0,c,\gamma) = u(f_0,c)+\gamma g(f_0,c),
\end{align}
where $\gamma = (\gamma(t))_{t \in [0,T]}$ is a Lagrangian multiplier function.
By Pontryagin's principle, the optimal solution 
$c^\ast = (c^\ast(t))_{t \in [0,T]}$ must satisfy
\begin{align*}
    c^\ast(t) = {\rm argmin}_c {\cal H}(f_0,c,\gamma)
\end{align*}
for each $t \in [0,T]$, 
where $f_0 = (f_0(t))_{t \in [0,T]}$ and $(\gamma(t))_{t \in [0,T]}$ satisfy the differential equations
\begin{align*}
        & \frac{\partial {\cal H}}{\partial \gamma} = f_0', \\
    & \frac{\partial {\cal H}}{\partial f_0} = -\gamma'.
\end{align*}
The boundary conditions are $f_0(0) = \alpha$ for some fixed $\alpha \in [0,1]$ and $\gamma(T) = 0$.
For our problem, the two differenial equations above become
\begin{align}
    & 
    f_0' = 
     (\lambda_1-\lambda_0(c))f_0^2 - (\lambda_1-\lambda_0(c)+\mu(c)+\nu(c))f_0+\nu(c), \quad f_0(0) = \alpha, \label{eq:forwardode} \\
         & 
    \gamma' = 
    -\big(2(\lambda_1-\lambda_0(c))f_0-(\lambda_1-\lambda_0(c)+\mu(c)+\nu(c))\big)\gamma
    -(\lambda_0(c)-\lambda_1), \quad \gamma(T) = 0. \label{eq:backwardode}
\end{align}
To compute the optimal solution, we apply the forward-backward sweep method, which is an iterative procedure that proceeds as follows:
\begin{enumerate}[(1)]
\item We discretize time and consider a set of time points ${\cal T} = [0,t_1,t_2,\ldots,T]$.
    \item We make an initial guess for the optimal policy $(c(t))_{t \in {\cal T}}$. We simply use $c(t) = 0$ for all $t \in {\cal T}$.
    \item We solve \eqref{eq:forwardode} using the initial condition $f_0(0)=\alpha$, assuming the initial guess for $(c(t))_{t \in {\cal T}}$.
\item We solve \eqref{eq:backwardode} {\em backwards in time} using the initial condition $\gamma(T)=0$, assuming the initial guess for $(c(t))_{t \in {\cal T}}$ and the trajectory for $(f_0(t))_{t \in {\cal T}}$ from step (3).
\item Using the obtained trajectories for $(f_0(t))$ and $(\gamma(t))$, for each fixed $t \in {\cal T}$ we compute the optimal dose at time $t$ as
\begin{align*}
    c^\ast(t) = {\rm argmax}_c {\cal H}(f_0,c,\gamma),
\end{align*}
which boils down to solving
\begin{align} \label{eq:optimalitycondition3}
    \frac{\partial {\cal H}}{\partial c} = 0.
\end{align}
For our problem, this condition can be written as a second-degree polynomial equation in $c$.
Since
\begin{align*}
    \frac{\partial u}{\partial c} &= f_0 \frac{\partial \lambda_0}{\partial c}, \\
    \frac{\partial g}{\partial c} &= 
    -f_0^2 \frac{\partial \lambda_0}{\partial c} +
    f_0 \frac{\partial \lambda_0}{\partial c} - 
    f_0 \frac{\partial \mu}{\partial c} -
    f_0 \frac{\partial \nu}{\partial c} +
    \frac{\partial \nu}{\partial c},
\end{align*}
we have
\begin{align*}
    \frac{\partial \cal H}{\partial c} &= 
    \frac{\partial u}{\partial c} + \gamma \frac{\partial g}{\partial c} \\ &= 
    (f_0 - \gamma f_0^2 + \gamma f_0) \frac{\partial \lambda_0}{\partial c} - \gamma f_0 \frac{\partial \mu}{\partial c} + \gamma (1 - f_0) \frac{\partial \nu}{\partial c}.
\end{align*}
Now,
\begin{align*}
    \lambda_0(c) = \lambda_0 - \Delta d_0 \cdot \frac{c}{c+1} &\quad\Rightarrow\quad \frac{\partial \lambda_0}{\partial c} = - \frac{\Delta d_0}{(c+1)^2}, \\
    \mu(c) = \mu_0 + kc &\quad\Rightarrow\quad \frac{\partial \mu}{\partial c} = k, \\
     \nu(c) = \nu_0 - mc &\quad\Rightarrow\quad \frac{\partial \nu}{\partial c} = -m.
\end{align*}
If we set
\begin{align*}
    & a(f_0,\gamma) := \Delta d_0 f_0  (\gamma f_0-1-\gamma), \\
    & b(f_0,\gamma) := - \gamma k f_0 + \gamma m (f_0-1), \\
\end{align*}
we then have
\begin{align*}
    \frac{\partial \cal H}{\partial c} = \frac{a}{(c + 1)^2} + b,
\end{align*}
and we obtain the second-degree polynomial equation
\begin{align*}
    bc^2 + 2bc + (a + b) = 0.
\end{align*}
\item Steps (3)--(5) are repeated until convergence. 
\end{enumerate}
Figure \ref{fig:optimallinear} in the main text shows optimal dosing strategies under linear induction of tolerance computed using our implementation of the forward-backward sweep method.
Theser results are consistent with computations obtained using the Dymos optimal control library in Python together with the IPOPT solver, as shown in Figure \ref{fig:dymos}.

\subsection{Mathematical analysis of optimal long-run treatment} \label{app:proofs}

In this section, we wish to prove two claims made in the main text:
\begin{itemize}
    \item For linear induction of tolerance, it is optimal in the long run to apply the dose which minimizes the equilibrium growth rate $\sigma(c)$ under a constant dose $c$.
    \item For uniformly induced tolerance, arbitrarily fast alternation between the maximum dose $c_{\rm max}$ and no dose minimizes the long-run average growth rate.
\end{itemize}
We do this in several steps in the following four subsections.

\subsubsection{Constant-parameter models}

We first establish a simple and intuitive lemma saying that 
a lower proliferation potential for drug-sensitive cells implies a lower equilibrium growth rate.
In addition, we establish a lower bound on the equilibrium growth rate.
In what follows, we will let $M(\lambda_0,\lambda_1,\mu,\nu)$ denote the constant-parameter model with parameters $\lambda_0$, $\lambda_1$, $\mu$ and $\nu$.

\begin{lemma} \label{lemma0}
In the constant-parameter model $M(\lambda_0,\lambda_1,\mu,\nu)$, the equilibrium growth rate $\sigma = \sigma(\lambda_0,\lambda_1,\mu,\nu)$ is increasing in $\lambda_0$ and it satisfies $\sigma > \lambda_1-\nu$.
\end{lemma}

\begin{proof}
The equilibrium growth rate is given by the expression
\begin{align*}
     \sigma = \frac12 \left((\lambda_0-\mu) + (\lambda_1-\nu) + \sqrt{((\lambda_0-\mu)-(\lambda_1-\nu))^2+4\mu\nu}\right).
\end{align*}
Consider the function
\begin{align*}
    f(x) := \frac12 \left(x-a+\sqrt{(x-b)^2+c}\right), \quad x \in \mathbb{R},
\end{align*}
where
\begin{align*}
    &a := \mu-\lambda_1+\nu, \quad b := \mu+\lambda_1-\nu, \quad c := 4\mu\nu.
\end{align*}
The derivative of this function is
\begin{align*}
    f'(x) = \frac12\left(1 + \frac{x-b}{\sqrt{(x-b)^2+c}}\right), \quad x \in \mathbb{R}.
\end{align*} 
Since we assume $\mu,\nu>0$ (Sections \ref{sec:model} and \ref{sec:drugeffectswitching}), we have
\begin{align*}
    |x-b| < \sqrt{(x-b)^2+c},
\end{align*}
which implies that $f'(x)>0$ for all $x \in \mathbb{R}$ and $f$ is therefore strictly increasing.
Since also
\begin{align*}
    \lim_{x \to -\infty} f(x) = \frac12(b - a) = \lambda_1-\nu,
\end{align*}
we can conclude that $\sigma > \lambda_1-\nu$.
\end{proof}

We next establish a simple lemma relating the equilibrium growth rate of a constant-parameter model to the instantaneous growth rate $u(c,f_0) = (\lambda_0-\lambda_1)f_0+\lambda_1$.

\begin{lemma} \label{lemma1}
   Let $M(\lambda_0,\lambda_1,\mu,\nu)$ be a constant-parameter model with equilibrium growth rate $\sigma = \sigma(\lambda_0,\lambda_1,\mu,\nu)$ and equilibrium proportion  of sensitive cells $\bar{f}_0 = \bar{f}_0(\lambda_0,\lambda_1,\mu,\nu)$. Then
   \[
   (\lambda_0-\lambda_1) \bar{f}_0 + \lambda_1 = \sigma.
   \]
\end{lemma}

\begin{proof}
If the model is started in equilibrium, meaning that $n_0(0) = N \bar{f}_0$ and $n_1(0) = N (1-\bar{f}_0)$ for some $N>0$, the population grows exponentially at rate $\sigma$ from the start and the proportion of sensitive cells remains fixed at $\bar{f}_0$ throughout.
This can be verified by consulting the explicit expression provided in Section \ref{sec:constantdoseequilibrium} of the main text.
However, from expression \eqref{eq:dndtsuppl}, we see that in general, if the current proportion of sensitive cells is $f_0$, the instantaneous growth rate is $(\lambda_0-\lambda_1)f_0 + \lambda_1$. 
Thus, we must have $(\lambda_0-\lambda_1)\bar{f}_0 + \lambda_1 = \sigma$.
\end{proof}

\subsubsection{Supporting results for general dosing functions $\lambda_0(c)$, $\mu(c)$ and $\nu(c)$}

We next establish two supporting lemmas which hold for arbitrary functions $\lambda_0(c)$, $\mu(c)$ and $\nu(c)$ that are integrable on every finite interval $[0,T]$. We additionally assume that $\lambda_0(c)$ is decreasing in $c$ with $\lambda_0(0) = \lambda_0$ and that $\mu(c)$ and $\nu(c)$ are bounded by  $\mu_0 \leq \mu(c) \leq \mu_{\rm max}$ and $\nu_{\rm min} \leq \nu(c) \leq \nu_0$ for $c \in [0,c_{\rm max}]$ with $\mu_0,\nu_{\rm min} > 0$.

Before proceeding,
we need to define a few terms.
For a given infinite-horizon treatment $c: [0,\infty) \to [0,c_{\rm max}]$, we define a {\em treatment segment} to be the restriction $c|_I: I \to [0,c_{\rm max}]$ of $c$ to $I$, where $I \subseteq [0,\infty)$ can be written as a finite disjoint union of closed intervals
\[
I = \bigcup_{i=1}^n [a_i,b_i]
\]
 with $a_i \leq b_i$ for all $i=1,\ldots,n$, and
\begin{align} \label{eq:connectingsegments}
f_0(b_i) = f_0(a_{i+1}), \quad i=1,\ldots,n-1.    
\end{align}
In other words, $c|_I$ is a treatment on the bounded set $I$ which starts with the sensitive cell proportion $f_0(a_1)$.
We also define the average growth rate of the segment,
\[
\bar{u}_I(c) = \frac1{|I|} \int_I \left((\lambda_0(c)-\lambda_1)f_0+\lambda_1\right) dt,
\]
where $|I|$ is the total length of the intervals that comprise $I$ (the Lebesgue measure of $I$).
We say that the treatment segment is {\em $\varepsilon$-stable} if 
\begin{align} \label{eq:epsilonstable}
\sup_{x,y \in I} |f_0(x)-f_0(y)| \leq \varepsilon. 
\end{align}
We now show that if a treatment segment is $\varepsilon$-stable, its average growth rate can be approximated by the equilibrium growth rate of a constant-parameter model.
The constant parameters are obtained by taking averages of the model parameters over the segment.

\begin{lemma} \label{lemma2}
Let $\varepsilon>0$
and let $c|_I: I \to [0,c_{\rm max}]$ be an $\varepsilon$-stable treatment segment with length $|I| \geq 1$ and average growth rate $\bar{u}_I(c)$.
Let $\bar{\lambda}_0(I)$, $\bar{\mu}(I)$ and $\bar{\nu}(I)$ denote the average values of the respective parameters over $I$:
\begin{align*}
    & \bar{\lambda}_0(I) = \frac1{|I|} \int_I \lambda_0(c)dt, \quad \bar{\mu}(I) = \frac1{|I|} \int_I \mu(c)dt, \quad \bar{\nu}(I) = \frac1{|I|} \int_I \nu(c)dt,
\end{align*}
and let $\sigma = \sigma(\bar\lambda_0(I),\lambda_1,\bar\mu(I),\bar\nu(I))$ denote  the equilibrium growth rate of the constant-parameter model $M(\bar{\lambda}_0(I), \lambda_1, \bar{\mu}(I), \bar{\nu}(I))$.
Then, there exist constants $C,D>0$ which only depend on $\lambda_0$, $\lambda_1$, $\mu_0$, $\mu_{\rm max}$, $\nu_0$ and $\nu_{\rm min}$ so that 
\[
\varepsilon \leq C \quad\Rightarrow\quad \sigma \leq \bar{u}_I(c)+ D\varepsilon.
\]
\end{lemma}

\begin{proof}
Write $I = \bigcup_{i=1}^n [a_i,b_i]$ and let $g := f_0(a_1)$ be the starting proportion of sensitive cells on the treatment segment.
Set
\begin{align*}
    \alpha(c) := \lambda_1 - \lambda_0(c), \quad \beta(c) := -\lambda_1 + \lambda_0(c) - \mu(c) - \nu(c).
\end{align*}
By \eqref{eq:df0ftsuppl}, we have that
$$
    f_0' = 
     \alpha(c) f_0^2 + \beta(c) f_0+\nu(c),
$$
which implies by \eqref{eq:connectingsegments},
\begin{align} \label{eq:integraldifferential}
\left| \int_I \left(\alpha(c)f_0^2 + \beta(c)f_0 + \nu(c)\right) dt \right| = \left| \sum_{i=1}^n (f_0(b_i)-f_0(a_i)) \right| = \left|f(b_n)-f(a_1)\right| \leq \varepsilon.
\end{align}
Now, let
\begin{align*}
    \bar{\alpha}(I) := \frac1{|I|} \int_I \alpha(c)dt, \quad \bar{\beta}(I) := \frac1{|I|} \int_I \beta(c)dt,
\end{align*}
be the averages values of $\alpha(c)$ and $\beta(c)$ over $I$, where we note that
\begin{align*}
    \frac1{|I|} \int_I |\bar{\alpha}(c)| \leq \lambda_0+\lambda_1, \quad \frac1{|I|} \int_I |\bar{\beta}(c)| \leq \lambda_0+\lambda_1+\mu_{\rm max}+\nu_0.
\end{align*}
In what follows, we write $\bar\alpha = \bar\alpha(I)$, $\bar\beta = \bar\beta(I)$ and $\bar\nu = \bar\nu(I)$ for ease of notation.
Note that
\begin{align*}
    & \bar{\alpha}g^2 + \bar{\beta}g + \bar{\nu} \\
     =& \frac{1}{|I|} \int_I \left(\alpha(c)g^2 + \beta(c)g + \nu(c)\right) dt \\
     =& \frac{1}{|I|} \int_I \left(\alpha(c)f_0^2 + \beta(c)f_0 + \nu(c)\right) dt +  \frac{1}{|I|} \int_I  \left(\alpha(c)(g^2-f_0^2) + \beta(c)(g - f_0)\right) dt,
\end{align*}
which implies by \eqref{eq:integraldifferential}, the fact that $c|_I$ is $\varepsilon$-stable and the assumption $|I| \geq 1$,
\begin{align} \label{eq:Pg}
    \left|  \bar{\alpha}g^2 + \bar{\beta}g + \bar{\nu} \right| &\leq \frac\varepsilon{|I|} + \varepsilon \left(\frac2{|I|} \int_I |\bar{\alpha}(c)|+ \frac1{|I|} \int_I |\bar{\beta}(c)|\right) \leq 
    C_1\varepsilon, 
\end{align}
where $C_1 :=  1+3(\lambda_0+\lambda_1)+\mu_{\rm max}+\nu_0 >0$.
Now, let  $\bar{f}_0 = \bar{f}_0(\bar{\lambda}_0(I), \lambda_1, \bar{\mu}(I), \bar{\nu}(I))$ be the equilibrium proportion of sensitive cells under the constant-parameter model $M(\bar{\lambda}_0(I), \lambda_1, \bar{\mu}(I), \bar{\nu}(I))$.
By Lemma \ref{lemma1},
\begin{align*}
    \sigma = (\bar{\lambda}_0(I)-\lambda_1) \bar{f}_0 + \lambda_1.
\end{align*}
Set $P(x) := \bar{\alpha}x^2 + \bar{\beta}x + \bar{\nu}$ and note that $P(\bar{f}_0)=0$.
Since $P(0)= \bar{\nu}>0$ and $P(1) = -\bar{\mu}<0$, $\bar{f}_0$ is the unique root of $P$ in $[0,1]$.
Moreover, $|P(g)| \leq C_1\varepsilon$ by \eqref{eq:Pg}, which suggests that $g$ and $\bar{f}_0$ are close to each another when $\varepsilon$ is small.
In fact, by applying a linear approximation of $P$ around $\bar{f}_0$, we can show that the distance between $g$ and $\bar{f}_0$ is of order $\varepsilon$.
More precisely, there exist constants $C_2,C_3>0$ 
so that if $\varepsilon \leq C_2$, then
\begin{align*}
    |g-\bar{f}_0| \leq C_3\varepsilon,
\end{align*}
where the constants $C_2,C_3>0$ only depend on 
$\lambda_0$, $\lambda_1$, $\mu_0$, $\mu_{\rm max}$, $\nu_0$ and $\nu_{\rm min}$.
Since this statement is intuitive and the proof involves straightforward calculations,
we defer the proof to Section \ref{app:technicallemma}.
Using the statement, we obtain since $c|_I$ is $\varepsilon$-stable,
\begin{align*}
\bar{u}_I(c) &= \frac1{|I|} \int_I \left((\lambda_0(c)-\lambda_1)f_0+\lambda_1\right) dt \\
&= \frac1{|I|} \int_I \left((\lambda_0 (c) - \lambda_1)g + \lambda_1\right) dt + \frac1{|I|} \int_I (\lambda_0 (c) - \lambda_1)(f_0 - g) dt \\
&= (\bar{\lambda}_0(I) - {\lambda}_1)g + {\lambda}_1 + O(\varepsilon) \\
&= (\bar{\lambda}_0(I) - {\lambda}_1)\bar{f}_0 + {\lambda}_1 + O(\varepsilon) \\
&= \sigma + O(\varepsilon),
\end{align*}
which is sufficient to prove the result.
\end{proof}

We next show that for any treatment $c: [0,\infty) \to [0,c_{\rm max}]$ satisfying a certain regularity condition, we can find an $\varepsilon$-stable treatment segment $c|_I$ with approximately the same average growth rate as the long-run average growth rate $\bar{u}_\infty(c)$ of $c$,
given by \eqref{eq:longrunavgratesuppl}.
The regularity condition is imposed to exclude treatments that make infinitely many dose changes over a finite time interval.
The condition is for example satisfied if there exists a $\delta>0$ so that each selected dose is given for at least $\delta$ time units, which is a reasonable condition in practice.

\begin{lemma} \label{lemma3}
Let $c: [0,\infty) \to [0,c_{\rm max}]$ be an infinite-horizon treatment so that the associated sensitive cell proportion trajectory satisfies the following regularity condition:
\begin{equation}
  \tag{A}\label{eq:A}
  \parbox{\dimexpr\linewidth-4em}{%
    \strut
    For each closed interval $[d,e] \subseteq [0,1]$ and each $T \geq 0$, $f_0^{-1}([d,e]) \cap [0,T]$ is a finite union of disjoint closed intervals $[a_i,b_i]$ where possibly $a_i=b_i$.
    \strut
  }
\end{equation}
Then, for each $\varepsilon>0$, there exists an $\varepsilon$-stable treatment segment $c|_I$ so that
\[
|I| \geq 1 \quad\text{and}\quad \bar{u}_I(c) \leq \bar{u}_\infty(c) + \varepsilon.
\]
\end{lemma}

\begin{proof}
We begin by selecting $d_0,d_1,d_2,\ldots,d_n \in [0,1]$ so that $d_0 = 0$, $d_n=1$ and $0<d_{k}-d_{k-1} \leq \varepsilon$ for all $k = 1,\ldots,n$.
By assumption \eqref{eq:A}, each set $I_k(T) := f_0^{-1}([d_{k-1},d_{k}]) \cap [0,T]$ for $k=1,\ldots,n$ and $T \geq 0$ is a finite union of disjoint closed intervals.
Note that we can write
\begin{align*}
    \bar{u}_T(c) = \bar{u}_{[0,T]}(c) = \sum_{k=1}^n \frac{|I_k(T)|}{T} \bar{u}_{I_k(T)}(c),
\end{align*}
with the understanding that $\frac{|I_k(T)|}{T} \bar{u}_{I_k(T)}(c) = 0$ when $|I_k(T)|=0$.
Now let $K := \{k \in \{1,\ldots,n\}: |I_k(T)| \geq 1 \text{ for some $T>0$}\}$ be the index set of the segments that eventually reach length 1. 
Note that $\bar{u}_{I_k(T)}(c)$ is bounded by $\lambda_0+2\lambda_1$ when $|I_k(T)| \leq 1$.
Therefore,
\begin{align*}
    &\lim_{T \to \infty} \sum_{k \in \{1,\ldots,n\} \setminus K} \frac{|I_k(T)|}{T} = 0, \quad \lim_{T \to \infty} \sum_{k \in K} \frac{|I_k(T)|}{T} = 1,\\
    &\lim_{T \to \infty} \sum_{k \in \{1,\ldots,n\} \setminus K} \frac{|I_k(T)|}{T} \bar{u}_{I_k(T)}(c) = 0.
\end{align*}
Now let $T_0$ be large enough so that $|I_k(T_0)| \geq 1$ for all $k \in K$, 
and so that for any $T \geq T_0$,
\begin{align} \label{eq:tocontradict}
    \left| \sum_{k \in \{1,\ldots,n\} \setminus K} \frac{|I_k(T)|}{T} \bar{u}_{I_k(T)}(c) \right| \leq \frac\varepsilon2.
\end{align}
Assume that $\bar{u}_{I_k(T)}(c) \geq \bar{u}_\infty(c) + {\varepsilon}$ for all $k \in K$ and $T \geq T_0$. Then for all $T \geq T_0$
\begin{align*}
    \sum_{k=1}^n \frac{|I_k(T)|}{T} \bar{u}_{I_k(T)}(c) \geq \sum_{k \in K} \frac{|I_k(T)|}{T} \bar{u}_{I_k(T)}(c) - \frac\varepsilon2 \geq (\bar{u}_\infty(c)+\varepsilon) \sum_{k \in K} \frac{|I_k(T)|}{T} - \frac{\varepsilon}2,
\end{align*}
which leads to a contradiction when we send 
$T \to \infty$, since
\begin{align*}
    \bar{u}_\infty(c) = \limsup_{T \to \infty} \bar{u}_T(c) = \limsup_{T \to \infty} \sum_{k=1}^n \frac{|I_k(T)|}{T} \bar{u}_{I_k(T)}(c).
\end{align*}
Therefore, there must exist $k$ and $T$ so that $|I_k(T)| \geq 1$ and
$\bar{u}_{I_k(T)}(c) \leq  \bar{u}_\infty(c) + {\varepsilon}.$
Now $c|_{I_k(T)}$
is the desired treatment segment.
\end{proof}

\subsubsection{Results for linear and uniform induction of tolerance}

We can now prove that for linear induction of tolerance, where $\mu(c) = \mu_0+kc$ for $k \geq 0$ and $\nu(c) = \nu_0 - mc$ for $m \geq 0$, a constant-dose treatment is optimal in the long run.

\begin{theorem} \label{thm:linearcase}
For linear induction of tolerance, 
then for any treatment $c: [0,\infty) \to [0,c_{\rm max}]$ which satisfies \eqref{eq:A}, we have that
\begin{align*}
    \min_{\tilde c \in [0,c_{\rm max}]} \sigma(\tilde c) \leq \bar{u}_\infty(c).
\end{align*}
\end{theorem}

\begin{proof}
    Let $0<\varepsilon\leq C$, where $C$ is chosen as in Lemma \ref{lemma2}.
    By Lemma \ref{lemma3}, we can find an $\varepsilon$-stable treatment segment $c|_I$ so that $|I| \geq 1$ and $\bar{u}_I(c) \leq \bar{u}_\infty(c)+\varepsilon$.
    By Lemma \ref{lemma2}, the equilibrium growth rate $\rho = \sigma(\bar{\lambda}_0(I), \lambda_1, \bar{\mu}(I), \bar{\nu}(I))$ of the constant-parameter model 
    $M(\bar{\lambda}_0(I), \lambda_1, \bar{\mu}(I), \bar{\nu}(I))$
    satisfies $\rho \leq \bar{u}_I(c) + D\varepsilon$, where the constant $D>0$ is independent of $\varepsilon$.
    Therefore,
    \begin{align*}
        \rho \leq \bar{u}_I(c) + D\varepsilon \leq \bar{u}_\infty(c) + (D+1)\varepsilon.
    \end{align*}
    Let $\bar{c}(I) := \frac1{|I|} \int_I c(t)dt$ denote the average dose over $I$.
    Since $\mu(c)$ and $\nu(c)$ are linear in $c$, $\mu(\bar{c}(I)) = \bar\mu(I)$ and $\nu(\bar{c}(I)) = \bar{\nu}(I)$.
    Furthermore, since $\lambda_0$ is convex, $\lambda_0(\bar{c}(I)) \leq \bar{\lambda}_0(I)$ by Jensen's inequality.
    Then, by Lemma \ref{lemma0},
    we have that $\sigma(\bar{c}(I)) \leq \rho$, where $\sigma(c)$ denotes the equilibrium growth rate under the constant-parameter model $M(\lambda_0(c),\lambda_1,\mu(c),\nu(c))$.
    We can conclude that
    \begin{align*}
        \min_{\tilde c \in [0,c_{\rm max}]} \sigma(\tilde c) \leq \sigma(\bar{c}(I)) \leq \bar{u}_\infty + (D+1)\varepsilon.
    \end{align*}
    By sending $\varepsilon \to 0$ in the right-hand side, we obtain the desired result.
\end{proof}

We can also prove that for uniformly induced tolerance, where $\mu(0) = \mu_0$, $\mu(c) = \mu_{\rm max}$ for $c>0$ and $\nu(0) = \nu_0$, $\nu(c) = \nu_{\rm min}$ for $c>0$, 
an arbitrarily fast pulsed schedule alternating between the maximum dose $c_{\rm max}$ and no dose is {optimal}.

\begin{theorem} \label{thm:uniformcase}
For uniform induction of tolerance,
then for any treatment $c: [0,\infty) \to [0,c_{\rm max}]$ which satisfies \eqref{eq:A}, we have that
\begin{align*}
    \min_{\varphi \in [0,1]} \sigma(c_{\rm max},\varphi) \leq \bar{u}_\infty(c).
\end{align*}
\end{theorem}

\begin{proof}
     Let $0<\varepsilon\leq C$, where $C$ is chosen as in Lemma \ref{lemma2}.
     Like in the proof of Theorem \ref{thm:linearcase}, we can find an $\varepsilon$-stable treatment segment $c|_I$ so that the constant-parameter model $M(\bar{\lambda}_0(I), \lambda_1, \bar{\mu}(I), \bar{\nu}(I))$ has equilibrium growth rate $\rho \leq \bar{u}_\infty + (D+1)\varepsilon$.
     Now set $S := \{t \in I: c(t) > 0\}$ and note that
     \begin{align*}
         \bar\mu(I) = \frac1{|I|} \int_I \mu(c) dc = \frac1{|I|} \left(\int_S \mu(c) dc + \int_{I \setminus S} \mu(c) dx\right) = \frac{|S|}{|I|} \mu_{\rm max} + \frac{|I \setminus S|}{|I|} \mu_0.
     \end{align*}
     Similar calculations hold for $\bar\nu$. Therefore, there exists $\varphi \in [0,1]$ so that $\bar{\mu}(I) = \varphi \mu_{\rm max} + (1-\varphi) \mu_0$ and $\bar{\nu}(I) = \varphi \nu_{\rm min} + (1-\varphi) \nu_0$.
     We also have that
    \begin{align} \label{eq:uniforminductionineq}
         \bar\lambda_0(I) &= \frac1{|I|} \left(\int_S \lambda_0(c) dc + \int_{I \setminus S} \lambda_0(c) dx\right) \nonumber \\
         &\geq \frac{|S|}{|I|} \lambda_0(c_{\rm max}) + \frac{|I \setminus S|}{|I|} \lambda_0 \nonumber \\
         &= \varphi \lambda_0(c_{\rm max}) + (1-\varphi) \lambda_0.
     \end{align}
    Recall from \eqref{eq:sigmapulsedsuppl} that $\sigma(c,\varphi)$ denotes the equilibrium growth rate of the constant-parameter model $M(\bar{\lambda}_0(c,\varphi), \lambda_1, \bar\mu(c,\varphi), \bar\nu(c,\varphi))$, where $\bar\lambda_0(c,\varphi) = \varphi \lambda_0(c) + (1-\varphi) \lambda_0$, $\bar\mu(c,\varphi) = \varphi \mu(c) + (1-\varphi) \mu_0$ and $\bar\nu(c,\varphi) = \varphi \nu(c) + (1-\varphi)\nu_0$.
    Now, by \eqref{eq:uniforminductionineq} and Lemma \ref{lemma0},
    \begin{align*}
        \sigma(c_{\rm max},\varphi)
        \leq \rho,
    \end{align*}
    from which it follows that
    \begin{align*}
        \min_{\varphi \in [0,1]} \sigma(c_{\rm max},\varphi) \leq \bar{u}_\infty + (D+1)\varepsilon.
    \end{align*}
    By sending $\varepsilon \to 0$ in the right-hand side, we obtain the desired result.
\end{proof}

\subsection{Completion of the proof of Lemma \ref{lemma2}} \label{app:technicallemma}

Here, we complete the proof of Lemma \ref{lemma2}.
We want to show that there exist constants $C_2,C_3>0$ 
so that if $\varepsilon \leq C_2$, then
\begin{align*}
    |g-\bar{f}_0| \leq C_3\varepsilon,
\end{align*}
where the $C_2$ and $C_3$ only depend on 
$\lambda_0$, $\lambda_1$, $\mu_0$, $\mu_{\rm max}$, $\nu_0$ and $\nu_{\rm min}$.
Recall that $P(x) = \bar\alpha x^2 + \bar\beta x + \bar\nu$, $\bar{f}_0$ is the unique root of $P$ in $[0,1]$, and $g \in [0,1]$ satisfies $|P(g)| \leq C_1\varepsilon$.

\begin{figure}
    \centering
    \includegraphics[scale=0.6]{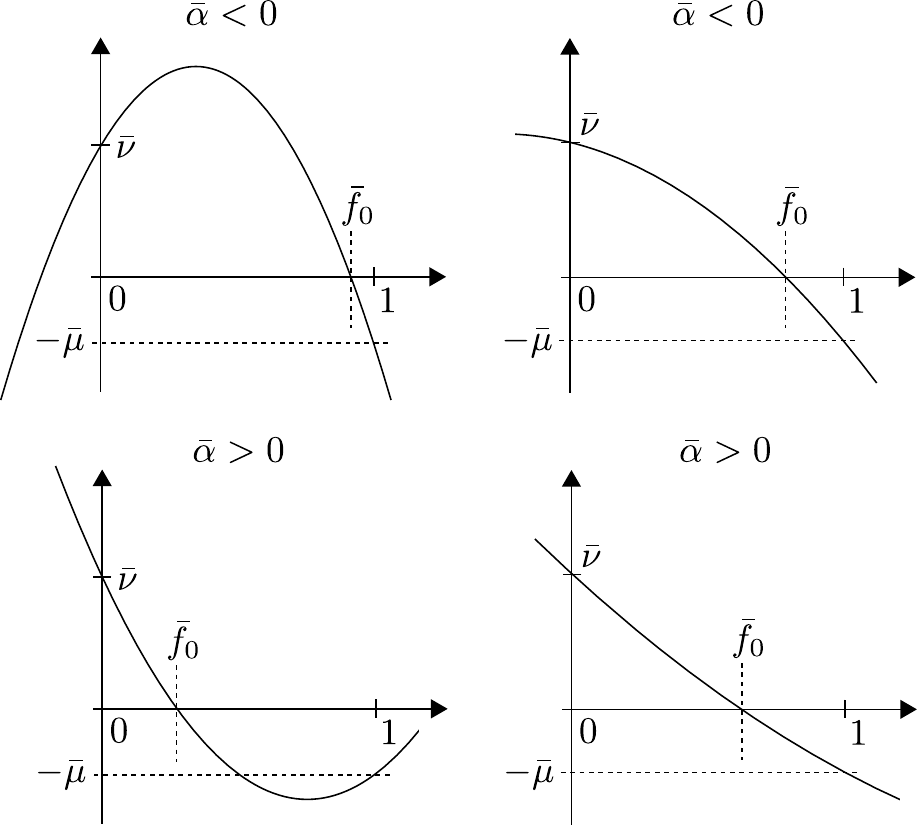}
    \caption{Potential shapes of the function $P(x)$ depending on whether $\bar{\alpha}>0$ or $\bar{\alpha}<0$.}
    \label{fig:2nddegree}
\end{figure}

For the case $\bar\alpha = 0$, $P(x)$ is linear in $x$ and the proof is straightforward.
If $\bar\alpha \neq 0$, 
Figure \ref{fig:2nddegree} shows
the four possible scenarios,
depending on the sign of $\bar\alpha$ and whether the axis of symmetry for $P$ falls within $[0,1]$.
In all cases, $P$ is locally strictly decreasing at $\bar{f}_0$ and the preimage $P^{-1}([-C_1\varepsilon,C_1\varepsilon])$ is an interval $[x_0,x_1]$ containing $g$ and $\bar{f}_0$, {assuming $\varepsilon$ is sufficiently small that $C_1\varepsilon \leq \bar\nu$}.
Now assume that $\bar\alpha<0$. Then, there are two solutions to $P(x)=0$ and $\bar{f}_0$ is the larger solution (Figure \ref{fig:2nddegree}). Therefore, since $\bar\alpha<0$,
\begin{align} \label{eq:barf0expr}
    \bar{f}_0 = \frac{-\bar\beta - \sqrt{(\bar \beta)^2-4\bar\alpha \bar\nu}}{2\bar\alpha},
\end{align}
and the endpoints $x_0,x_1$ are given by
\begin{align*}
    \begin{rcases*} x_0 \\ x_1\end{rcases*} = \frac{-\bar\beta - \sqrt{(\bar \beta)^2-4\bar\alpha(\bar\nu \mp C_1\varepsilon)}}{2\bar\alpha}.
\end{align*}
Consider first the left endpoint $x_0$, for which $x_0<\bar{f}_0$ and $P(x_0) = C_1\varepsilon$.
Note that
\begin{align*}
    \sqrt{(\bar \beta)^2-4\bar\alpha(\bar\nu - C_1\varepsilon)} &= \sqrt{(\bar \beta)^2-4\bar\alpha\bar\nu} \sqrt{1 - \frac{4|\bar\alpha|C_1}{\sqrt{(\bar \beta)^2-4\bar\alpha\bar\nu}} \varepsilon},
\end{align*}
where we use that $(\bar\beta)^2-4\bar\alpha\bar\nu>0$ since $\bar\alpha<0, \bar\nu>0$ and that $\bar\alpha = -|\bar\alpha|$.
We can now use the inequality 
$\sqrt{1-x} \geq 1-x$ for $0\leq x \leq 1$
to say that for
$\varepsilon \leq \sqrt{(\bar \beta)^2-4\bar\alpha\bar\nu}/(4|\bar\alpha|C_1)$,
\begin{align*}
    \sqrt{1 - \frac{4|\bar\alpha|C_1}{\sqrt{(\bar \beta)^2-4\bar\alpha\bar\nu}} \varepsilon} \geq 1 - \frac{4|\bar\alpha|C_1}{\sqrt{(\bar \beta)^2-4\bar\alpha\bar\nu}} \varepsilon,
\end{align*}
which implies since $\bar\alpha<0$,
\begin{align*}
&\bar{f}_0-x_0 \\
    &=\bar{f}_0 - \frac{-\bar\beta - \sqrt{(\bar \beta)^2-4\bar\alpha(\bar\nu-C_1\varepsilon)}}{2\bar\alpha} \\ 
    &= \bar{f}_0 + \frac{\bar\beta}{2\bar\alpha} + \frac{\sqrt{(\bar \beta)^2-4\bar\alpha\bar\nu} \sqrt{1 - \frac{4|\bar\alpha|C_1}{\sqrt{(\bar \beta)^2-4\bar\alpha\bar\nu}} \varepsilon}}{2\bar\alpha} \\
        &\leq \bar{f}_0 + \frac{\bar\beta}{2\bar\alpha} + \frac{\sqrt{(\bar \beta)^2-4\bar\alpha\bar\nu}}{2\bar\alpha} - \frac{\sqrt{(\bar \beta)^2-4\bar\alpha\bar\nu}}{2\bar\alpha} \cdot \frac{4|\bar\alpha|C_1}{\sqrt{(\bar \beta)^2-4\bar\alpha\bar\nu}} \varepsilon \\
    &\leq \bar{f}_0 + \frac{\bar\beta}{2\bar\alpha} + \frac{\sqrt{(\bar \beta)^2-4\bar\alpha\bar\nu}}{2\bar\alpha} + 2C_1 \varepsilon \\
    &=  2C_1 \varepsilon.
\end{align*}
where in the final step we use $\eqref{eq:barf0expr}$.
For the right endpoint $x_1$, for which $x_1>\bar{f}_0$ and $P(x_1) = -C_1\varepsilon$, we note similary that
\begin{align*}
    \sqrt{(\bar \beta)^2-4\bar\alpha(\bar\nu + C_1\varepsilon)} &= \sqrt{(\bar \beta)^2-4\bar\alpha\bar\nu} \sqrt{1 + \frac{4|\bar\alpha|C_1}{\sqrt{(\bar \beta)^2-4\bar\alpha\bar\nu}} \varepsilon}.
\end{align*}
We can now use the inequality $\sqrt{1+x} \leq 1+x$ for $x \geq 0$ to say that 
\begin{align*}
    \sqrt{1 + \frac{4|\bar\alpha|C_1}{\sqrt{(\bar \beta)^2-4\bar\alpha\bar\nu}} \varepsilon} \leq 1 + \frac{4|\bar\alpha|C_1}{\sqrt{(\bar \beta)^2-4\bar\alpha\bar\nu}} \varepsilon,
\end{align*}
which implies since $\bar\alpha<0$,
\begin{align*}
&x_1 - \bar{f}_0 \\
    &=\frac{-\bar\beta - \sqrt{(\bar \beta)^2-4\bar\alpha(\bar\nu+C_1\varepsilon)}}{2\bar\alpha} - \bar{f}_0 \\ 
    &= -\frac{\bar\beta}{2\bar\alpha} - \frac{\sqrt{(\bar \beta)^2-4\bar\alpha\bar\nu} \sqrt{1 + \frac{4|\bar\alpha|C_1}{\sqrt{(\bar \beta)^2-4\bar\alpha\bar\nu}} \varepsilon}}{2\bar\alpha} - \bar{f}_0 \\
    &\leq -\bar{f}_0 - \frac{\bar\beta}{2\bar\alpha} - \frac{\sqrt{(\bar \beta)^2-4\bar\alpha\bar\nu}}{2\bar\alpha} + 2C_1 \varepsilon \\
    &=  2C_1 \varepsilon.
\end{align*}
Since $g \in [x_0,x_1]$, we can now guarantee that for $\varepsilon \leq \sqrt{(\bar \beta)^2-4\bar\alpha\bar\nu}/(4|\bar\alpha|C_1)$,
\begin{align*}
    |g-\bar{f}_0| \leq 2C_1\varepsilon.
\end{align*}
The same statement applies to the case $\bar\alpha>0$.
We have therefore shown that there exist constants $C_2,C_3>0$ so that if $\varepsilon \leq C_2$, then $|g-\bar{f}_0| \leq C_3 \varepsilon$. 
It remains to be shown that the constants $C_2,C_3$ can be selected so that they only depend on
$\lambda_0$, $\lambda_1$, $\mu_0$, $\mu_{\rm max}$, $\nu_0$ and $\nu_{\rm min}$.
For $C_3$, this is easy, since we can set $C_3 := 2C_1$ and we know that $C_1 = 1+3(\lambda_0+\lambda_1)+(\mu_{\rm max}+\nu_0)$.
For $C_2$, note first that $|\bar\alpha| \leq \lambda_0+\lambda_1$ and note next that
\begin{align*}
    & (\bar \beta)^2-4\bar\alpha\bar\nu = \left(\lambda_1 - \bar\lambda_0 + \bar\mu + \bar\nu\right)^2-4(\lambda_1 - \bar\lambda_0)\bar\nu.
\end{align*}
Consider the polynomial $Q(x) := (x+\bar\mu+\bar\nu)^2-4x\bar\nu$.
This polynomial takes the smallest value at $x=\bar\nu-\bar\mu$, and the smallest value is $4\bar\mu\bar\nu \geq 4\mu_0\nu_{\rm min}$.
Recalling that we also need $C_1 \varepsilon \leq \bar\nu$, we set
\begin{align*}
    &C_2 := \frac1{1+3(\lambda_0+\lambda_1)+(\mu_{\rm max}+\nu_0)} \min\left\{\frac{2\sqrt{\mu_0\nu_{\rm min}}}{\lambda_0+\lambda_1},\nu_{\rm min}\right\}.
\end{align*}
Since $C_1 = 1+3(\lambda_0+\lambda_1)+(\mu_{\rm max}+\nu_0)$, we have that
\begin{align*}
    C_2 \leq \frac1{C_1} \min\left\{\frac{\sqrt{(\bar \beta)^2-4\bar\alpha\bar\nu}}{4|\bar\alpha|},\bar\nu\right\},
\end{align*}
as desired.

\subsection{Transient phase of the optimal treatment} \label{app:transient}

\subsubsection{General case}

Assume that at the start of treatment, the proportion of sensitive cells is $\bar{f}_0(0)$.
The objective of our optimal control problem
is to minimize \eqref{eq:optcontobjsuppl}.
Both for the case of linear and uniformly induced tolerance,
there is an optimal long-run growth rate $\sigma^\ast$
and an associated equilibrium proportion $\bar{f}_0^\ast$.
Instead of minimizing \eqref{eq:optcontobjsuppl}, we can minimize
\begin{align*}
    \int_0^T (u(c,f_0)-\sigma^\ast)dt,
\end{align*}
and since the growth rate under the optimal policy is eventually $\sigma^\ast$, we have that
\begin{align*}
    \int_0^T (u(c,f_0)-\sigma^\ast)dt = \int_0^\infty (u(c,f_0)-\sigma^\ast)dt,
\end{align*}
given that $T$ is sufficiently large.
During the transient phase, we can assume that $f_0(t)$ is strictly decreasing.
The reason is that if there are two times $t_1<t_2$ such that $f_0(t_1) = f_0(t_2)$ and $\int_{t_1}^{t_2} (u(c,f_0) - \sigma^\ast) dt < 0$, we can repeat the segment between $t_1$ and $t_2$ indefinitely to obtain a treatment with a long-run average growth rate less than $\sigma^\ast$, which is a contradiction.
If $\int_{t_1}^{t_2} (u(c,f_0) - \sigma^\ast) dt \geq 0$, we can remove the segment between $t_1$ and $t_2$ without increasing the integral $\int_0^\infty (u(c,f_0)-\sigma^\ast)dt$. 
Assuming that $f_0(t)$ is strictly decreasing, we can apply the following coordinate change:
\begin{align*}
    \int_0^\infty (u(c(t),f_0(t))-\sigma^\ast)dt &= \int_{\bar{f}_0(0)}^{\bar{f}_0^\ast} \frac{u(c, f_0) - \sigma^\ast}{{df_0}/{dt}}df_0 = \int_{\bar{f}_0^\ast}^{\bar{f}_0(0)} \frac{u(c, f_0) - \sigma^\ast}{-{df_0}/{dt}}df_0.
\end{align*}
To minimize this integral, we can for each fixed $f_0$ minimize
\begin{align*}
    \frac{u(c, f_0) - \sigma^\ast}{-f_0'}
\end{align*}
with respect to $c$, or equivalently maximize
\begin{align*}
    \frac{\sigma^\ast-u(c, f_0)}{-f_0'}
\end{align*}
with respect to $c$. 
Since $f_0(t)$ is strictly decreasing, we can apply the constraint that $f_0' < 0$.

\subsubsection{Linear induction of tolerance}

For linear induction of tolerance, we wish to reformulate the condition $f_0'<0$ in terms of $c$. By \eqref{eq:df0ftsuppl}, we have that
\begin{align*}
        f_0' &= \big(\lambda_1-\lambda_0(c)\big)f_0^2 - \big(\lambda_1-\lambda_0(c)+\mu(c)+\nu(c)\big)f_0+\nu(c),
\end{align*}
which we can write explicitly as
\begin{align*}
    f_0' &= \left(\lambda_1-\lambda_0+\Delta d_0 \frac{c}{c+1}\right) f_0^2 - \left(\lambda_1 -\lambda_0+\Delta d_0 \frac{c}{c+1} + \mu_0+kc + \nu_0-mc\right)f_0 + \nu_0-mc \\
    &= \Delta d_0 f_0(f_0-1) \frac{c}{c+1} + (mf_0-kf_0-m)c + \left((\lambda_1-\lambda_0)f_0(f_0-1) - (\mu_0+\nu_0)f_0+\nu_0\right)\\
    &= \frac{\Delta d_0 f_0(f_0-1) c + (mf_0-kf_0-m)c(c+1) + \left((\lambda_1-\lambda_0)f_0(f_0-1) - (\mu_0+\nu_0)f_0+\nu_0\right)(c+1)}{c+1}.
\end{align*}
Note that
\begin{align*}
    &\Delta d_0 f_0(f_0-1) c + (mf_0-kf_0-m)c(c+1) + \left((\lambda_1-\lambda_0)f_0(f_0-1) - (\mu_0+\nu_0)f_0+\nu_0\right)(c+1) \\
    &= (m(f_0-1)-kf_0)c^2 \\
    &\quad+ \left(m(f_0-1)-kf_0+\Delta d_0 f_0(f_0-1) + (\lambda_1-\lambda_0)f_0(f_0-1) - (\mu_0+\nu_0)f_0+\nu_0\right)c \\
    &\quad+ \left((\lambda_1-\lambda_0)f_0(f_0-1) - (\mu_0+\nu_0)f_0+\nu_0
    \right).
\end{align*}
If we now set 
\begin{align*}
    & A := m(f_0-1)-kf_0, \\
    & B := \Delta d_0 f_0(f_0-1), \\
    & D := (\lambda_1-\lambda_0)f_0(f_0-1) - (\mu_0+\nu_0)f_0+\nu_0, 
\end{align*}
we can write
\begin{align*}
    & f_0' = \frac{Ac^2 + (A+B+D)c + D}{c+1}.
\end{align*}
The equation $Ac^2 + (A+B+D)c+D=0$ has solutions
\begin{align*}
    c = \frac{-(A+B+D) \pm \sqrt{(A+B+D)^2-4AD}}{2A}.
\end{align*}
Since $k,m \geq 0$, with either $k \neq 0$ or $m \neq 0$, and $0 < f_0 < 1$, we have $A < 0$, and for $f_0 \leq \bar{f}_0(0)$, we have $D \geq 0$. This implies that $AD \leq 0$. 
If $D=0$, we get that $f_0'<0$ if and only if $Ac^2+(A+B)c<0$, which is equivalent to $c>-(A+B)/A$ with the restriction that $c \geq 0$. 
In this case, we take $c_{\rm min}(f_0) := -(A+B)/A$.
If $D>0$, then $Ac^2+(A+B+D)c+D=0$ has one positive and one negative solution. Let $c_{\rm min}(f_0)$ denote the positive solution.
Then with the restriction that $c \geq 0$, the condition $f_0'<0$ is equivalent to $c > c_{\rm min}(f_0)$.

\subsubsection{Uniform induction of tolerance}

Here, we want to show that under uniform induction of tolerance, the maximum dose $c_{\rm max}$ is optimal during the transient phase.
First, we note that if the treatment is started at $\bar{f}_0(0)$, we can exclude the possibility that $c=0$ is optimal since this will lead to $f_0' \geq 0$.
It therefore suffices to consider doses $c>0$.
Using \eqref{eq:df0ftsuppl} and \eqref{eq:instgrowthratesuppl}, we can explicitly write
\begin{align*} 
    \frac{\sigma^\ast-u(c,f_0)}{-f_0'} &= \frac{\sigma^\ast - (\lambda_0(c)-\lambda_1)f_0 - \lambda_1}{\big(\lambda_0(c)-\lambda_1\big)f_0^2 + \big(\lambda_1-\lambda_0(c)+\mu(c)+\nu(c)\big)f_0 - \nu(c)} \\
    &= \frac{\lambda_0(c) f_0 + \big(\lambda_1(1-f_0) - \sigma^\ast\big)}{\lambda_0(c) f_0(1-f_0) + \big(\nu(c)(1-f_0) - \mu(c) f_0 - \lambda_1 f_0 (1-f_0)\big)}
\end{align*}
In this case, we have $\mu(c) = \mu_{\rm max}$ and $\nu(c) = \nu_{\rm min}$ for $c>0$.
Therefore, if we set
\begin{align*}
    & a := f_0, \\
    & b := \lambda_1(1-f_0) - \sigma^\ast, \\
    & d := f_0(1-f_0), \\
    & e := \nu_{\rm min}(1-f_0) - \mu_{\rm max} f_0 - \lambda_1f_0(1-f_0),
\end{align*}
we can write for $c>0$:
\begin{align} \label{eq:transientuniformlyinduced}
    \frac{\sigma^\ast-u(c,f_0)}{-f_0'} &= \frac{a\lambda_0(c)+b}{d\lambda_0(c)+e}.
\end{align}
This leads us to studying the function
\begin{align*}
    & f(x) := \frac{ax+b}{dx+e}.
\end{align*}
Since $f'(x) = (ae-bd)/(dx+e)^2$, this function is either strictly increasing, strictly decreasing or constant depending on the sign of $ae-bd$.
We wish to show that for any $f_0 > \bar{f}_0^\ast$, we have $ae-bd<0$, meaning that $f(x)$ is strictly decreasing in $x$.
This in turns means that the largest value of \eqref{eq:transientuniformlyinduced} is attained when $\lambda_0(c)$ is smallest, which occurs at $c = c_{\rm max}$.

To analyze $ae-bd$, we begin by computing:
\begin{align*}
    ae-bd &= f_0 \big(\nu_{\rm min}(1-f_0) - \mu_{\rm max} f_0 - \lambda_1f_0(1-f_0)\big) - f_0(1-f_0) \big(\lambda_1(1-f_0) - \sigma^\ast\big) \\
    &= \nu_{\rm min} f_0 (1-f_0) - \mu_{\rm max} f_0^2 - \lambda_1 f_0^2 (1-f_0) - \lambda_1 f_0 (1-f_0)^2 + \sigma^\ast f_0(1-f_0) \\
    &= \nu_{\rm min} f_0 (1-f_0) - \mu_{\rm max} f_0^2 - \lambda_1 f_0 (1-f_0) + \sigma^\ast f_0(1-f_0) \\
    &= \big(\sigma^\ast-\lambda_1+\nu_{\rm min}\big) f_0(1-f_0) - \mu_{\rm max} f_0^2.
\end{align*}
If $\sigma^\ast-\lambda_1+\nu_{\rm min} \leq 0$, then $ae-bd \leq -\mu_{\rm max} f_0^2 < 0$ whenever $f_0 > \bar{f}_0^\ast$ and we are done. For the case $\sigma^\ast-\lambda_1+\nu_{\rm min} > 0$, we define 
\begin{align*}
    & m := \sigma^\ast-\lambda_1+\nu_{\rm min}, \\
    & n := \mu_{\rm max}
\end{align*}
with $m,n > 0$, and we further define the function
\begin{align*}
    g(x) := mx(1-x)-nx^2 = -(m+n)x^2+mx.
\end{align*}
Note that $g$ is a second-degree polynomial with roots $x=0$ and $x = \frac{m}{m+n}$.
Since $m+n>0$, 
we know
that $g(x) > 0$ for $x \in [0,m/(m+n))$ and $g(x) < 0$ for $x \in (m/(m+n),1]$.
If we are able to show that
\begin{align} \label{eq:toshow}
    \bar{f}_0^\ast \geq \frac{m}{m+n} = \frac{\sigma^\ast-\lambda_1+\nu_{\rm min}}{\sigma^\ast-\lambda_1+\mu_{\rm max}+\nu_{\rm min}},
\end{align}
it will follow that for any $f_0 > \bar{f}_0^\ast$, we have $ae-bd<0$ as desired.

To establish \eqref{eq:toshow}, note first that by Lemma \ref{lemma1},
\begin{align*}
    \frac{\sigma^\ast-\lambda_1+\nu_{\rm min}}{\sigma^\ast-\lambda_1+\mu_{\rm max}+\nu_{\rm min}} &= \frac{(\bar{\lambda}_0(c_{\rm max},\varphi^\ast)-\lambda_1) \bar{f}_0^\ast + \nu_{\rm min}}{(\bar{\lambda}_0(c_{\rm max},\varphi^\ast)-\lambda_1)\bar{f}_0^\ast+\mu_{\rm max}+\nu_{\rm min}},
\end{align*}
where
\begin{align*}
    \bar\lambda_0(c,\varphi) = \varphi \lambda_0(c) + (1-\varphi)\lambda_0,
\end{align*}
as defined in \eqref{eq:avgratessuppl}.
Next note that by \eqref{eq:df0ftrapidpulsesuppl},
\begin{align*} 
    & (\bar{\lambda}_0(c_{\rm max},\varphi^\ast)-\lambda_1)((\bar{f}_0^\ast)^2-\bar{f}_0^\ast) + \bar\mu(c_{\rm max},\varphi^\ast) \bar{f}_0^\ast -\bar\nu(c_{\rm max},\varphi^\ast) (1-\bar{f}_0^\ast) \\
    &= \big(\bar\lambda_0(c_{\rm max},\varphi^\ast)-\lambda_1\big)(\bar{f}_0^\ast)^2 - \big(\bar\lambda_0(c_{\rm max},\varphi^\ast)-\lambda_1-\bar\mu(c_{\rm max},\varphi^\ast)-\bar\nu(c_{\rm max},\varphi^\ast)\big)\bar{f}_0^\ast-\bar\nu(c_{\rm max},\varphi^\ast) \\
    &= 0.
\end{align*}
Since
\begin{align*}
    &  \bar\mu(c_{\rm max},\varphi^\ast) = \varphi^\ast \mu_{\rm max} + (1-\varphi^\ast)\mu_0 = \mu_{\rm max} + (1-\varphi^\ast)(\mu_0-\mu_{\rm max}), \\
    &  \bar\nu(c_{\rm max},\varphi^\ast) = \varphi^\ast \nu_{\rm min} + (1-\varphi^\ast)\nu_0 = \nu_{\rm min} + (1-\varphi^\ast)(\nu_0-\nu_{\rm min}),
\end{align*}
it follows that
\begin{align*}
    & (\bar{\lambda}_0(c_{\rm max},\varphi^\ast)-\lambda_1)((\bar{f}_0^\ast)^2-\bar{f}_0^\ast) + \mu_{\rm max} \bar{f}_0^\ast -\nu_{\rm min} (1-\bar{f}_0^\ast) \\
    &= (1-\varphi^\ast)(\mu_{\rm max}-\mu_0) \bar{f}_0^\ast + (1-\varphi^\ast)(\nu_0-\nu_{\rm min}) (1-\bar{f}_0^\ast) \\
    &\geq 0,
\end{align*}
which implies that
\begin{align*}
    \frac{(\bar{\lambda}_0(c_{\rm max},\varphi^\ast)-\lambda_1) \bar{f}_0^\ast + \nu_{\rm min}}{(\bar{\lambda}_0(c_{\rm max},\varphi^\ast)-\lambda_1)\bar{f}_0^\ast+\mu_{\rm max}+\nu_{\rm min}} \leq \bar{f}_0^\ast.
\end{align*}
This concludes the proof.

\subsection{Pulsed schedules} \label{app:pulsed}

\subsubsection{More rapid pulses give at least as low long-run growth rate}

We show that any pulsed treatment $T_1 := (c, t_{\rm cycle}, \varphi_1)$ can be modified to get an equally good or better pulsed treatment $T_2 := (c, t_{\rm cycle} / 2, \varphi_2)$. We assume $T_1$ is started so that $f_0$ is already in equilibrium, i.e., $f_0$ is a periodic function of time with period $t_{\rm cycle}$. Assume the treatment is started in such a way that the drug is applied from time $t = 0$ to time $t = \varphi_1 t_{\rm cycle}$ and then not applied from time $t = \varphi_1t_{\rm cycle}$ to time $t_{\rm cycle}$. Then, $f_0$ starts at some $y_1$ at time $t = 0$, changes monotonically to $y_2$ over $\varphi_1t_{\rm cycle}$ units of time, and then changes monotonically back to $y_1$ over $(1 - \varphi_1)t_{\rm cycle}$ units of time. There exists some $y_m$ between $y_1$ and $y_2$ such that the combined time it takes for $f_0$ to go from $y_1$ to $y_m$ and from $y_m$ to $y_1$ is exactly $t_{\rm cycle} / 2$. Then the combined time for $f_0$ to go from $y_m$ to $y_2$ and from $y_2$ to $y_m$ is exactly $t_{\rm cycle} / 2$ as well. Let $t_1$ and $t_2$ ($t_1,t_2 < t_{\rm cycle})$ be times in the first period where $f_0 = y_m$. Since $T_1$ is started in equilibrium, long-run average growth rate is simply the average growth rate over the time-interval $[0, t_{\rm cycle}]$, which is in turn the average of the growth rates over the sets $[0, t_1] \cup [t_2, t_{\rm cycle}]$ and $[t_1, t_2]$. The average growth rate over one of these sets must be at least as good as the long-run growth rate. If the former is better then we can simply choose $T_2$ to be the treatment where we apply dose $c$ for time $t_1$ and then don't apply the dose from time $t_2$, that is, $T_2 = (c, t_{\rm cycle} / 2, 2t_1 / t_{\rm cycle})$. If the latter is better we can similarly chose $T_2$ to be the treatment where we apply the drug for time $\varphi t_{\rm cycle} - t_1$ and do not apply the drug for time $t_2 - \varphi t_{\rm cycle}$, that is $T_2 = (c, t_{\rm cycle} / 2, 2\varphi - 2t_1/t_{\rm cycle})$.

\subsubsection{Once pulses are rapid enough, the long-run growth rate is insensitive to the cycle length}

Here we show that if the treatment pulses are fast enough for the graph of $f_0$ to be approximately linear during both the on-drug and off-drug phase of each treatment cycle, then the cycle duration does not matter as long as the average $f_0$ is the same.

Assume we have two different pulsed dosing schedules $T_1 = (c, t_{{\rm cycle},1}, \varphi)$ and $T_2 = (c, t_{{\rm cycle},2}, \varphi) $ that are applied in a scenario where they have the same average value of $f_0$ and in which $f_0$ has reached equilibrium, meaning that $f_0$ is periodic. Additionally assume that $f_0$ behaves linearly during both the on-drug and off-drug phase in each case.
Since the lines are approximations of the true model, we should assume the corresponding line segments have the same slope. Let $y_1$ and $y_2$ be the lowest and largest values of $f_0$ under schedule $T_1$, and similarly, let $y_3$ and $y_4$ be the smallest and largest values under schedule $T_2$. Let  $[t_1, t_3]$ be a period of the treatment with schedule $T_1$ such that the drug is administered  during the interval $[t_1, t_2]$ and is not administered during the interval $[t_2, t_3]$. Under the above assumptions, the average growth rate using schedule $T_1$ is
\begin{align*}
\varphi\frac{\int_{t_1}^{t_2} ((\lambda_0(c) - \lambda_1)f_0 + \lambda_1) \ dt}{t_2 - t_1} + (1 - \varphi)\frac{\int_{t_2}^{t_3} ((\lambda_0(0) - \lambda_1)f_0 + \lambda_1) \ dt}{t_3 - t_2}
\end{align*}
but since $f_0$ is approximately linear on these intervals this is is equal to
\begin{align*}
 &   \varphi \left((\lambda_0(c) - \lambda_1)\frac{(t_2 - t_1)y_m}{t_2 - t_1} + \lambda_1\right)
+
(1 -\varphi) \left((\lambda_0(0) - \lambda_1)\frac{(t_3 - t_2)y_m}{t_3 - t_2} + \lambda_1\right) \\
&= \varphi(\lambda_0(c) - \lambda_1)y_m + (1 - \varphi)(\lambda_0(0) - \lambda_1)y_m + \lambda_1,
\end{align*}
which would be the same if we did the same calculations for $T_2$.

\subsection{Implementation of Russo et al.~model} \label{app:russo}

\subsubsection{Model and parametrization}

In Russo et al.~\citesupp{russo2022modified}, the authors use their experimental data to parametrize a mathematical model involving two cell types, sensitive and tolerant. 
The drug effect on the sensitive cell death rate $d_0$ is modeled using a function of the form
\begin{align*}
    d_0(c) = d_0 + (d_{\rm max}-d_0) (1-\exp(-ac)) = d_0 + \Delta d_0 (1-\exp(-ac)),
\end{align*}
where $c$ is the drug dose measured in $M$.
Under this model, the dose which has half the maximal effect on the sensitive cell death rate is given by ${\rm EC}_{50}^d = \log(2)/a$.
If we measure the drug dose as a proportion of
the ${\rm EC}_{50}^d$ dose, we can rewrite the above function as
\begin{align} \label{eq:drugeffectrusso}
    d_0(c) = d_0 + \Delta d_0(1-\exp(-c\log(2))).
\end{align}
We note that this version of $d_0(c)$ is concave in $c$, 
same as the functional form we assume in Section \ref{sec:drugeffectprolif} of the main text.

For the WiDr cell line, Russo et al.~infer that the transition rate $\mu(c)$ from sensitivity to tolerance is given by the linear function
\begin{align*}
    \mu(c) = k'c,
\end{align*}
where $c$ is the drug dose measured in $\mu M$.
If we again measure the drug dose as a proportion of the ${\rm EC}_{50}^d$ dose and define $k := k' {\rm EC}_{50}$, we can rewrite this function as
\begin{align*}
    \mu(c) = kc.
\end{align*}
With this modeling setup, the authors in \citesupp{russo2022modified} use their experimental data to derive the parameter estimates $\lambda_0 = b_0-d_0 = 0.048$, $a = 2.91495 \cdot 10^6$,  $\Delta d_0 = 1.095$, $\lambda_1-\nu_0 = -0.073$ and $k' = 136935$ for the WiDr B7 clone.
For these estimates, the ${\rm EC}_{50}^d$ dose is given by
\[
{\rm EC}_{50}^d = \frac{\log(2)}{2.91495 \cdot 10^6} = 2.3779 \cdot 10^{-7}\;M,
\]
and the rescaled slope $k$ is given by
\[
k = 136935 \cdot {\rm EC}_{50}^d = 0.03256.
\]

\subsubsection{Implementation of forward-backward sweep method}

To solve the optimal control problem minimizing \eqref{eq:optcontobjsuppl} 
using the drug effect function in \eqref{eq:drugeffectrusso}, we can apply the forward-backward sweep method as laid out in Section \ref{app:fbsm} with one modification.
In the fifth step, we have to solve $\partial {\cal H}/\partial c= 0$.
We again have that
\[
\frac{\partial {\cal H}}{\partial c} = (f_0-\gamma f_0^2+\gamma f_0) \frac{\partial \lambda_0}{\partial c} - \gamma f_0 \frac{\partial \mu}{\partial c} + \gamma (1-f_0) \frac{\partial \nu}{\partial c},
\]
but we must modify our calculation of $\partial \lambda_0/\partial c$. In this case, we have 
\begin{align*}
    \lambda_0(c) = \lambda_0 - \Delta d_0(1-\exp(-c\log(2))),
\end{align*}
which implies that
\[
\frac{\partial \lambda_0}{\partial c} = - \log(2) (\Delta d_0) \exp(-c\log(2)),
\]
and we obtain the equation
\[
 - \log(2) (\Delta d_0) \exp(-c\log(2)) (f_0-\gamma f_0^2+\gamma f_0) - \gamma f_0 k = 0.
\]
If we now set
\begin{align*}
& a(f_0,\gamma) := \log(2) (\Delta d_0) f_0 (\gamma f_0 - 1 - \gamma), \\
& b(f_0,\gamma) := -\gamma f_0 k,
\end{align*}
we obtain
\[
\exp(-c\log(2)) = -\frac{b}a,
\]
which yields $c = -\log(-b/a)/\log(2)$. This is the solution we use in the fifth step.

\subsubsection{Natural transitions between phenotypes}

In Russo et al.~\citesupp{russo2022modified}, the authors assume 
no transitions from sensitivity to tolerance 
in the absence of drug, $\mu_0=0$.
For the WiDr cell line (clone B7), their point estimate for the proportion of sensitive cells at the start of treatment is 0.992, which would indicate transitions between phenotypes before treatment.
However, the proportion is indistinguishable from zero when statistical uncertainty is taken into account.
The authors point out that transitions from tolerance to sensitivity may occur even in the presence of drug, but that these transitions would effectively contribute to the death rate of tolerant cells.
Overall, it is uncertain to what extent cells transition between sensitivity and tolerance in the absence of drug.

To be consistent with our modeling framework, we assume transitions between types do occur in the absence of drug, so that they lead to an equilibrium proportion $\bar{f}_0(0) = 0.992$ of sensitive cells in the absence of drug.
This leads us to take $\mu_0 = 0.001$, $\lambda_1=0$ and $\nu_0 = 0.073$, 
which is still consistent with Russo et al.'s finding that drug tolerance is primarily induced by the drug and that $\lambda_1-\nu_0 = -0.073$.
A summary of the parameter values used for the analysis in Section \ref{sec:russo} of the main text is given in Table \ref{table:parametersrusso}.

\begin{table}[h]
\centering
\begin{tabular}{|c|c|c|c|c|c|c|c|}
\hline
Parameter&$\lambda_0$ & $\Delta d_0$& $\lambda_1$& $\mu_0$& $\nu_0$&$k$ \\
\hline
Value&0.048&1.095&0&0.001&0.073&0.03256 \\
\hline
\end{tabular}
\caption{Parameter values used in Section \ref{sec:russo} of the main text.}
\label{table:parametersrusso}
\end{table}

\subsection{Connection with previous work}
\label{sec:previouswork}

In this section, we discuss the connection of our work with previous work, which is reviewed in more detail in \citesupp{gunnarsson2025}.
Kuosmanen et al.~\citesupp{kuosmanen2021drug} studied optimal dosing under linearly induced irreversible transitions $\mu(c)$ from sensitivity to resistance, 
and they sought to minimize the probability of resistance forming under treatment given an initially drug-sensitive tumor.
They showed it 
was best to administer constant low doses even under a modest level of drug-induced resistance, which is consistent with our results.
Greene et al.~\citesupp{greene2019mathematical,greene2020mathematical} considered a Lotka-Volterra competition model between sensitive and resistant cells, where they assumed that both $d_0(c)$ and $\mu(c)$ were linear in $c$, with transitions from sensitivity to resistance once again being irreversible.
In \citesupp{greene2019mathematical}, they showed that intermittent treatment can be preferable to constant treatment under drug-induced resistance, although this comparison was based on a single possible dose value.
In \citesupp{greene2020mathematical}, the authors studied the same model using optimal control theory, where the goal was to maximally delay tumor recurrence.
The optimal dosing strategies varied significantly with the specific model parameter values, involving periods of maximum dosing, no dosing, and small constant dosing.
We note that Greene et al.'s assumption of competition between sensitive and resistant cells fundamentally changes the model dynamics.
Under this assumption, it becomes counterproductive to apply large drug doses even when the drug does not induce resistance, since aggressive strategies release resistant cells from competition with sensitive cells.
Investigating how these dynamics may influence our optimal dosing strategies is an important direction for future inquiry.

In Angelini et al.~\citesupp{angelini2022model}, the authors studied a model of phenotypic switching where both 
$d_0(c)$ and $\mu(c)$
were of the form $c \mapsto k(1+e^{-r(c-s)})^{-1}$, which has a sigmoidal shape.
They only considered constant-dose schedules.
Their results indicate that when ${\rm EC}_{50}^d \approx {\rm EC}_{50}^\mu$, meaning that the ${\rm EC}_{50}$ dose is similar for $d_0(c)$ and $\mu(c)$, applying the maximum dose $c_{\rm max}$ is optimal.
However, when ${\rm EC}_{50}^\mu \gg {\rm EC}_{50}^d$, lower doses become optimal.
This is consistent with our results, since under the case ${\rm EC}_{50}^\mu \gg {\rm EC}_{50}^d$, $\mu(c)$ can be approximated by a linear function in $c$, as discussed in Section \ref{sec:drugeffectswitching}.
We have also shown that when ${\rm EC}_{50}^\mu \ll {\rm EC}_{50}^d$, in which case $\mu(c)$ can be considered approximately constant in $c$ for $c>0$, it is optimal to apply the maximum dose continuously.
In combination, our work and \citesupp{angelini2022model} indicate that drug-induced tolerance does not lead to optimality of small doses in and of itself.
Our work furthermore shows that if the drug inhibits the transition rate $\nu(c)$ from tolerance to sensitivity, applying the maximum dose continuously is generally suboptimal, irrespective of the relationship between ${\rm EC}_{50}^d$ and ${\rm EC}_{50}^\nu$.
In general, 
our work illustrates the importance of understanding the exact nature of drug-induced resistance before treatment recommendations can be made.

One of the most similar works to ours
is Akhmetzhanov et al.~\citesupp{akhmetzhanov2019modelling}, where the authors study a model of resistance in BRAF-mutant melanoma,
involving two mutually inhibitory biological pathways.
Their baseline model has three cell types, corresponding to an active main pathway (type-1), active alternative pathway (type-2), or both pathways inactive (type-0), with reversible transitions $1 \leftrightarrow 0 \leftrightarrow 2$.
With some simplifying assumptions, the authors reduce their model to a two-type model involving sensitive and resistant cells, where the rates $d_0(c)$, $\mu(c)$ and $\nu(c)$ are dose-dependent.
However, 
the transitions between states emerge from an underlying model of a particle undergoing Brownian motion inside a double-well potential, 
which results in $d_0(c)$, $\mu(c)$ and $\nu(c)$ having more complex forms than we have assumed.
Nevertheless, Akhmetzhanov et al.~\citesupp{akhmetzhanov2019modelling} arrive at an optimal strategy similar to the one we derived for linearly induced tolerance (Cases II and III), where large initial doses are followed by a small constant dose.
The authors note that the constant dose appears to maintain the population at an optimal composition and that the goal of the optimal treatment appears to be to reach this composition as quickly as possible.
Our work shows that the composition is optimal in the sense that it minimizes the equilibrium tumor growth rate $\sigma(c)$, and it offers a simple way of computing the equilibrium dose through minimization of $\sigma(c)$.
It also shows that the optimal strategy does not necessarily steer the tumor into the desired long-run composition as quickly as possible.
This is especially true when the drug only influences transitions from sensitivity to resistance, in which case it is optimal to apply a small dose throughout.
Rather, the optimal strategy must balance the trade-off between cell kill and tolerance induction in a way we have precisely quantified.

It is finally worth mentioning a recent work by Cassidy et al.~\citesupp{cassidy2021role}, where it is assumed that the probability of switching between a sensitive and tolerant state depends on the cell's age.
The authors do not assume that the drug influences the switching dynamics,
but rather that drug-tolerant cells cooperate to divide faster
with increasing frequency in the population.
Similar to the model of Greene et al., it is counterproductive to maximally kill the drug-sensitive cells in this case, as it triggers cooperation between the tolerant cells. 
In the context of non-small cell lung cancer \citesupp{craig2019cooperative},
Cassidy et al.~show that an adaptive therapy
aimed at limiting the frequency of tolerant cells in the population
outperforms a pulsed schedule applying a fixed dose every 7 days.
While the authors do not investigate optimal dosing strategies, 
it is interesting to note that their adaptive therapy
ultimately maintains the tumor 
around a certain
composition between sensitive and tolerant cells.
The principle of an optimal composition between sensitive and tolerant cells may therefore continue to apply even if our model is extended to account for potential 
group behavior.

\bibliographystylesupp{unsrt}
\bibliographysupp{epi}

\appendix

\makeatletter
\renewcommand\thefigure{S\@arabic\c@figure}

\setcounter{figure}{0}  
\setcounter{table}{0}  

\clearpage
\thispagestyle{empty}

\begin{figure}
    \centering
    \includegraphics[scale=1]{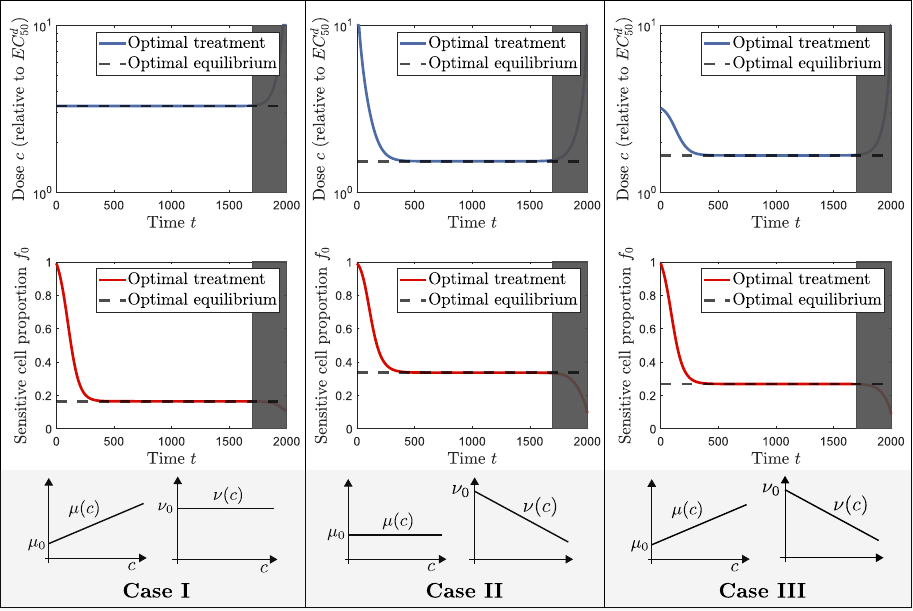}
    \caption{
    {\bf Effect of extending treatment horizon on optimal dosing strategies.}
    When the treatment horizon in Figure \ref{fig:optimallinear} is extended, the period during which the constant dose $c^\ast$ is applied becomes extended.
    }
    \label{fig:S1}
\end{figure}

\begin{figure}
    \centering
    \includegraphics[scale=1]{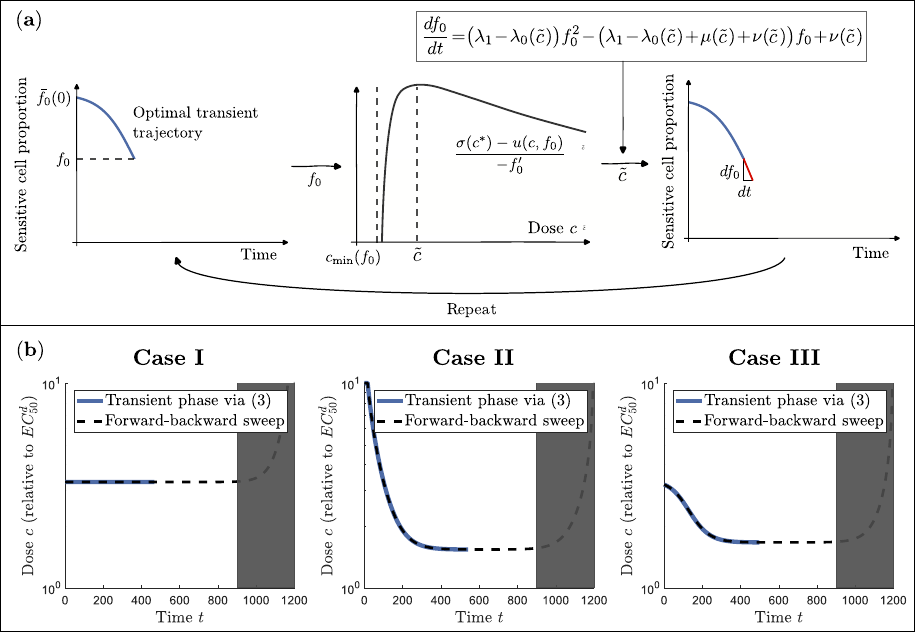}
    \caption{{\bf Understanding the transient phase of the optimal treatment.}
    {\bf (a)} 
    During the transient phase, whenever the proportion of sensitive cells is $f_0$, the optimal dose to give at that moment is the one maximizing the ratio $(\sigma(c^\ast)-u(c,f_0))/(-f_0')$, where $u(c,f_0)$ is the instantaneous tumor growth rate.
    This ratio precisely quantifies how the optimal treatment must balance cell kill and tolerance induction during the initial stages of treatment.
    It also gives a simple method for computing the transient phase, where 
    the ratio $(\sigma(c^\ast)-u(c,f_0))/(-f_0')$ is maximized and the value for $f_0$ is updated iteratively.
    {\bf (b)} The transient phase computed using the approach laid out in (a) (blue lines) agrees with the transient phase computed using the forward-backward sweep method (dotted lines).
    }
    \label{fig:transient_steering}
\end{figure}

\begin{figure}
    \centering
    \includegraphics[scale=1]{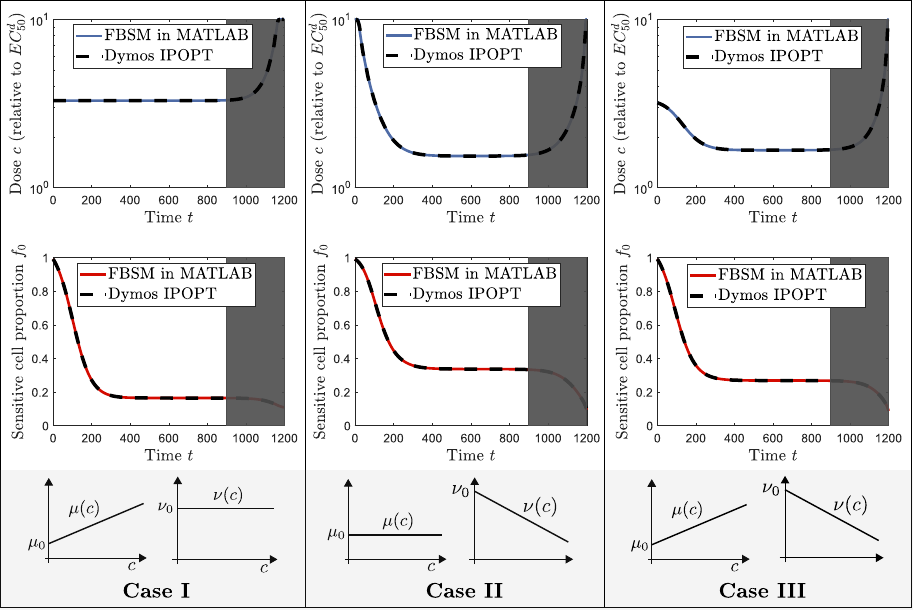}
    \caption{
    {\bf Optimal dosing strategies obtained using the Dymos optimal control library in Python.}
    The optimal dosing strategies computed using our implementation of the forward-backward sweep method (FBSM) in MATLAB are consistent with computations obtained using the Dymos optimal control library in Python (which is built on top of the 
    OpenMDAO optimization framework) together with the IPOPT solver.
    The Python code used to perform the Dymos computations is available in the Github respository for the paper.
    }
    \label{fig:dymos}
\end{figure}

\end{document}